\documentclass[a4paper,11pt]{article}

\usepackage{ifthen}

\newcommand{\MyMacro}[4]{
\provideboolean{#1}
\setboolean{#1}{#2}
\providecommand{#3}{\ifthenelse{\boolean{#1}}}
\providecommand{#4}{\ifthenelse{\not \boolean{#1}}}
}

\MyMacro{is_slides}{false}{\Ifslides}{\IfNotslides}

\MyMacro{use_packages}{true}{\IfUsepackages}{\IfNotUsepackages}

\MyMacro{plain_numbering}{true}{\IfPlainNumbering}{\IfNotPlainNumbering}

\MyMacro{single_numbering}{true}{\IfSingleNumbering}{\IfNotSingleNumbering}

\makeatletter
\MyMacro{is_llncs}{false}{\Ifllncs}{\IfNotllncs}

\makeatother

\newcommand{\ifndef}[2]{\ifthenelse{\isundefined{#1}}{#2}{}}

\newcommand{\mydef}[2]{\def#1{#2}}

\newcommand{\nospell}[1]{#1}   %

\makeatletter
\newcommand{\myusepackage}[2][]{\@ifpackageloaded{#2}{} %
{\ifthenelse{\equal{}{#1}} {\usepackage{#2}} {\usepackage[#1]{#2}} }}
\makeatother

\myusepackage{amssymb}
\myusepackage{amsmath}   %

\myusepackage{amsthm}   %

{}
\myusepackage{latexsym}
\myusepackage{amsfonts}
\myusepackage{units}     %
\myusepackage{psfrag}    %
\myusepackage{url}    %
\myusepackage{hyphenat}   %
\myusepackage{microtype}    %
{}
{}
\myusepackage[dvips]{graphicx}
\myusepackage[usenames,dvipsnames]{color}
{}
{}
{}

{}
\ifndef{\theorem}{\newtheorem{theorem}{Theorem}}
{}

{}   %
{}
\ifndef{\lemma}{\newtheorem{lemma}[theorem]{Lemma}}
\ifndef{\corollary}{\newtheorem{corollary}[theorem]{Corollary}}
\ifndef{\conjecture}{\newtheorem{conjecture}[theorem]{Conjecture}}
\ifndef{\remark}{\theoremstyle{remark} \newtheorem{remark}[theorem]{Remark}}
\ifndef{\proposition}{\newtheorem{proposition}[theorem]{Proposition}}
\ifndef{\claim}{\newtheorem{claim}[theorem]{Claim}}
\ifndef{\result}{\newtheorem{result}[theorem]{Result}}
\ifndef{\problem}{\newtheorem{problem}[theorem]{Problem}}
{}   %
{}   %

\newtheoremstyle{mydefinition}   %
{\topsep}{\topsep}   %
{\slshape}   %
{}   %
{\bfseries}   %
{.}   %
{ }   %
{}   %

\newtheoremstyle{myremark}   %
{\topsep}{\topsep}   %
{\slshape}   %
{}   %
{\bfseries\slshape}   %
{:}   %
{ }   %
{}   %

\newtheoremstyle{myexample}   %
{\topsep}{\topsep}   %
{\itshape}   %
{}   %
{\slshape}   %
{:}   %
{ }   %
{\ul{\thmname{#1}}}   %

\newtheoremstyle{myclaims}   %
{\topsep}{\topsep}   %
{\slshape}   %
{}   %
{\bfseries\itshape}   %
{}   %
{ }   %
{\thmname{#1}\thmnumber{ \!#2}.}   %

{}  %
{}  %
{\theoremstyle{myremark}}
{}
\ifndef{\definition}
{\theoremstyle{mydefinition}\newtheorem{definition}[theorem]{Definition}}
\ifndef{\example}
{\theoremstyle{myexample}\newtheorem{example}[theorem]{Example}}
{\theoremstyle{myclaims}

\ifndef{\fact}{\newtheorem{fact}[theorem]{Fact}}
}
{}  %
{}  %
{}  %

{}
\newtheoremstyle{anystatement}{\topsep}{\topsep}{\itshape}{}{\bfseries}{.}{ }{\anystatementname}
{\theoremstyle{anystatement}}
\newcommand{\anystatementname}{}

{}

\newcommand{\AuxNew}[4][]{#2{#3}[1][*]%
{\ifthenelse{\equal{*}{##1}} %
{\Ensuremath{#1{#4}}}%
{\ifthenelse{\equal{b}{##1}} %
{\Ensuremath{\mathbf{#4}}}%
{\ifthenelse{\equal{}{##1}} %
{\IfMathMode{#1{#4}}{#4}}{}}}}}

\newcommand{\newident}[3][*]{\ifthenelse{\equal{*}{#1}}%
{\AuxNew[\mathit]{\newcommand}{#2}{#3}} %
{\mydef{#2}{\Ensuremath{\mathit{#3}}}}} %

\newcommand{\newmat}[3][*]{\ifthenelse{\equal{*}{#1}}%
{\AuxNew{\newcommand}{#2}{#3}} %
{\mydef{#2}{\Ensuremath{#3}}}} %

\newcommand{\providemat}[3][*]{\ifthenelse{\equal{*}{#1}} %
{\AuxNew{\providecommand}{#2}{#3}} %
{\mydef{#2}{\Ensuremath{#3}}}} %

\newcommand{\newfunction}[2]{ %
\newcommand{#1}[2][*]{\ifthenelse{\equal{*}{##1}}%
{\Ensuremath{#2{\left(##2\right)}}}%
{#2(##2)}}%
}

\newcommand{\MyMakeTheoMacros}[3]{
\newcommand{#2}[2][]{\ifthenelse{\equal{}{##1}}
{\begin{#1} ##2 \end{#1}}
{\begin{#1}\label{##1} ##2\end{#1}}}
\newcommand{#3}[3][]{\ifthenelse{\equal{}{##1}}
{\begin{#1}[##2] ##3 \end{#1}}
{\begin{#1}[##2]\label{##1} ##3\end{#1}}}
}

\makeatletter
\newtheorem*{rep@theorem}{\rep@title}
\newcommand{\newreptheorem}[2]{%
\newenvironment{rep#1}[1]{%
\def\rep@title{#2 \ref{##1}}%
\begin{rep@theorem}}%
{\end{rep@theorem}}}
\makeatother

\newcommand{\MyMakeDupTheoMacros}[7]{
\MyMakeTheoMacros{#1}{#2}{#3}
\newreptheorem{#1}{#6}
\newcommand{#4}[3]{
\newcommand{##2}{##3}
\begin{#1}\label{##1} ##2\end{#1}}
\newcommand{#5}[4]{
\newcommand{##2}{##4}
\begin{#1}{\e{##3}}\label{##1} ##2\end{#1}}
\newcommand{#7}[2]{\begin{rep#1}{##1} ##2 \end{rep#1}}
}

\newcommand{\MyMakeRefMacros}[3]{\newcommand{#1}[2][]
{\ifthenelse{\equal{}{##1}}{#2~\ref{##2}}{#3~\ref{##1} and~\ref{##2}}}}

\newcommand{\MyMakeEqRefMacros}[3]{\newcommand{#1}[2][]
{\ifthenelse{\equal{}{##1}}{#2~\eqref{##2}}{#3~\eqref{##1} and~\eqref{##2}}}}

{}   %
\newcommand{\bibentry}[8]{
{}\bibitem[\nospell{#8}]{#1} {\textup #3}.{}
\ifthenelse{\equal{}{#6}}
{\newblock \textrm{#4.} \newblock {\em #5}, #7.}
{\newblock \textrm{#4.} \newblock {\em #5, #6}, #7.}
}

{}  %

{}  %
\MyMakeDupTheoMacros{lemma}
{\lem}{\nlem}{\lemdup}{\nlemdup}{Lemma}{\lemrep}
{}

\MyMakeRefMacros{\lemref}{Lemma}{Lemmas}

\MyMakeDupTheoMacros{corollary}
{\crl}{\ncrl}{\crldup}{\ncrldup}{Corollary}{\crlrep}

\MyMakeRefMacros{\crlref}{Corollary}{Corollaries}

\MyMakeTheoMacros{proposition}{\prp}{\nprp}

{}
\newtheorem*{prp*}{\e{Proposition}}

{}

\MyMakeRefMacros{\prpref}{Proposition}{Propositions}

\MyMakeDupTheoMacros{my_claim}
{\clm}{\nclm}{\clmdup}{\nclmdup}{Claim}{\clmrep}

\MyMakeRefMacros{\clmref}{Claim}{Claims}

\MyMakeDupTheoMacros{theorem}
{\theo}{\ntheo}{\theodup}{\ntheodup}{Theorem}{\theorep}

\MyMakeRefMacros{\theoref}{Theorem}{Theorems}

\MyMakeTheoMacros{definition}{\defi}{\ndefi}

\MyMakeRefMacros{\defiref}{Definition}{Definitions}

\MyMakeTheoMacros{problem}{\prob}{\nprob}

\MyMakeRefMacros{\probref}{Problem}{Problems}

\MyMakeTheoMacros{conjecture}{\conj}{\nconj}

\MyMakeRefMacros{\conjref}{Conjecture}{Conjectures}

\renewcommand{\qedsymbol}{$\blacksquare$}

\newcommand{\prfstart}[1][]{\ifthenelse{\equal{}{#1}} %
{\begin{proof}\renewcommand{\qedsymbol}{$\blacksquare$}} %
{\begin{proof}[Proof of #1] %
\renewcommand{\qedsymbol}{$\blacksquare_{\mbox{\it{\scriptsize{#1}}}}$}} %
}
\newcommand{\prfend}{\end{proof}\renewcommand{\qedsymbol}{$\blacksquare$}}

\newcommand{\ssect}[2][]{\ifthenelse{\equal{}{#1}}
{\subsection{#2}}
{\subsection{#2}\label{#1}}}

\MyMakeRefMacros{\chref}{Chapter}{Chapters}

\MyMakeRefMacros{\sref}{Section}{Sections}

\MyMakeRefMacros{\ssref}{Subsection}{Subsections}

\MyMakeRefMacros{\sssref}{Subsection}{Subsections}

\MyMakeRefMacros{\figref}{Figure}{Figures}

\newcommand{\IfMathMode}[2]{\ifmmode{#1}\else{#2}\fi}

\newcommand{\Ensuremath}{\ensuremath}

\newcommand{\fbr}[1]{\IfMathMode %
{#1}{$#1$}}                      %

\newcommand{\fnbr}[1]{\mbox{\fbr{#1}}}   %

\newcommand{\fla}[2][*]{\ifthenelse{\equal{}{#1}}{\fbr{#2}}{\fnbr{#2}}}

\newcommand{\mat}[2][]{\ifthenelse{\equal{}{#1}} %
{ \begin{displaymath} #2 \end{displaymath} } %
{ \begin{equation} \label{#1} #2 \end{equation} }%
}

\newcommand{\malabel}[1]{\addtocounter{equation}{1}\tag{\theequation}\label{#1}}
\newcommand{\mal}[2][]{ %
\ifthenelse{\equal{}{#1}} %
{{\begin{align*} #2 \end{align*}}}   %
{\ifthenelse{\equal{P}{#1}}                 %
{{\allowdisplaybreaks\begin{align*} #2            %
\end{align*}}} %
{{\begin{align*} \malabel{#1} #2 \end{align*}}}   %
} %
}

\newcommand{\f}{\fla}

\newcommand{\m}{\mat}

\MyMakeEqRefMacros{\equref}{Equation}{Equations}

\MyMakeEqRefMacros{\expref}{Expression}{Expressions}

\MyMakeEqRefMacros{\inequref}{Inequality}{Inequalities}

\newcommand{\bracref}[1]{(\ref{#1})}

\newcommand{\bref}{\bracref}

\providecommand{\middle}{\big}

\newcommand{\pl}[1][]{\nospell{\ifthenelse{\equal{}{#1}}%
{\!\stackrel'{}\!\!\txt{s}}%
{\fla{#1\!\stackrel'{}\!\!\txt{s}}}}}

\newcommand{\fr}[3][*]{%
\ifthenelse{\equal{*}{#1}}        %
{\frac{#2}{#3}}{}%
\ifthenelse{\equal{/}{#1}}        %
{\nicefrac{#2}{#3}}{}%
\ifthenelse{\equal{}{#1}}         %
{\left.#2\middle/#3\right.}{}%
\ifthenelse{\equal{p_}{#1}}       %
{\left.\left(#2\right)\middle/#3\right.}{}%
\ifthenelse{\equal{_p}{#1}}       %
{\left.#2\middle/\left(#3\right)\right.}{}%
\ifthenelse{\equal{pp}{#1}}       %
{\left.\left(#2\right)\middle/\left(#3\right)\right.}{}
}

\newcommand{\dr}{\nicefrac}

\newcommand{\sq}{\mathpalette\MySQRT}  %

\def\MySQRT#1#2{    %
\setbox0=\hbox{$#1\sqrt{#2\,}$}\dimen0=\ht0%
\advance\dimen0-0.2\ht0%
\setbox2=\hbox{\vrule height\ht0 depth -\dimen0}%
{\box0\lower0.4pt\box2}}

\newcommand{\set}[2][]{\ifthenelse{\equal{}{#1}} %
{\Ensuremath{\left\{#2\right\}}}%
{\Ensuremath{\left\{#2\middle|\vphantom{|_1^1}#1\right\}}}}

\newfunction{\asO}{O}

\mydef{\01}{\set{0,1}}

\newcommand{\lf}{\left}
\newcommand{\rt}{\right}

\newcommand{\sz}[2][]{\ifthenelse{\equal{}{#1}}%
{\Ensuremath{\left|#2\right|}}%
{\Ensuremath{\left|#2\middle|_{#1}\right.}}}

\newcommand{\txt}[1]{\textrm{#1}}   %

\DeclareMathAlphabet{\lowcal}{OT1}{pzc}{m}{it}

\newmat[]{\dt}{\cdot}
\newmat[]{\tm}{\cdot}
\newmat[]{\smin}{\setminus}
\newmat[]{\eps}{\varepsilon}
\newmat[]{\deq}{\stackrel{\textrm{def}}{=}}

\providemat{\QQ}{\mathbb{Q}}

\newcommand{\ds}[1][]
{\ifthenelse{\equal{}{#1}}{\dots}{#1\dots#1}}
\newmat[]{\dc}{\ds[,]}

\newcommand{\MyComment}[1]{\ClassWarning{My Macros}{#1}}

\newcommand{\fn}{\footnote}

{}  %
\newcommand{\e}{\emph}
{}   %

\newcommand{\bl}[1]{{\bf #1}}  %

\providecommand{\ul}[1]{\underline{#1}}  %

\MyComment{Look for ...-s}

\MyComment{Spell-check}

{}
{}

\title{Correlation in Hard Distributions in Communication Complexity}

\date{}

\newcommand{\instDG}{Institute of Mathematics, Czech Academy of Sciences, \v Zitna 25, Praha 1, Czech Republic.}
\newcommand{\thanksDG}{Partially funded by the grant P202/12/G061 of GA \v CR and by RVO:\ 67985840.
Most of this work was done while DG was visiting the Centre for Quantum Technologies at the National University of Singapore, and was partially funded by the Singapore Ministry of Education and the NRF.
}

\newcommand{\instHK}{Division of Mathematical Sciences, Nanyang Technological University, Singapore \& Centre for Quantum Technologies, National University of Singapore, Singapore.
}
\newcommand{\thanksHK}{This work is funded by the Singapore Ministry of Education (partly through the Academic Research Fund Tier 3 MOE2012-T3-1-009) and by the Singapore National Research Foundation.}

\newcommand{\instRB}{Division of Mathematical Sciences, Nanyang Technological University, Singapore.}

\IfNotllncs{
 \author{Ralph Bottesch\thanks{\instRB}
  \and Dmitry Gavinsky\thanks{\instDG\ \thanksDG}
  \and Hartmut Klauck\thanks{\instHK\ \thanksHK}
 }
}{
  \author[1]{Ralph C. Bottesch}
  \author[2]{Dmitry Gavinsky\footnote{\thanksDG}}
  \author[1,3]{Hartmut Klauck\footnote{\thanksHK}}
  \affil[1]{Nanyang Technological University\\
    50 Nanyang Avenue, Singapore 639798}
  \affil[2]{Institute of Mathematics, Czech Academy of Sciences\\
    \v Zitna 25, Praha 1, Czech Republic}
  \affil[3]{Centre for Quantum Technologies, National University of Singapore\\
    Block S15, 3 Science Drive 2, Singapore 117543}
  \authorrunning{R.\,C. Bottesch, D. Gavinsky, and H. Klauck}

  \Copyright{Ralph Christian Bottesch, Dmitry Gavinsky, and Hartmut Klauck}
  \subjclass{F.1.3}
  \keywords{communication complexity; information theory.}
}

\begin{document}

\maketitle

\begin{abstract}
We study the effect that the amount of correlation in a bipartite distribution has on the communication complexity of a problem under that distribution.
We introduce a new family of complexity measures that interpolates between the two previously studied extreme cases:\ the (standard) randomised communication complexity and the case of distributional complexity under product distributions.

We give a tight characterisation of the randomised complexity of Disjointness under distributions with mutual information $k$, showing that it is $\Theta(\sqrt{n(k+1)})$ for all $0\leq k\leq n$. This smoothly interpolates between the lower bounds of Babai, Frankl and Simon for the product distribution case ($k=0$), and the bound of Razborov for the randomised case.
The upper bounds improve and generalise what was known for product distributions, and imply that any tight bound for Disjointness needs $\Omega(n)$ bits of mutual information in the corresponding distribution.

We study the same question in the distributional {\em quantum} setting, and show a lower bound of $\Omega((n(k+1))^{1/4})$, and an upper bound (via constructing communication protocols), matching up to a logarithmic factor.

We show that there are total Boolean functions $f_d$ on $2n$ inputs that have distributional communication  complexity $O(\log n)$ under all distributions of information up to $o(n)$, while the (interactive) distributional complexity maximised over all distributions is $\Theta(\log d)$ for $6n\leq d\leq 2^{n/100}$.
This shows, in particular, that the correlation needed to show that a problem is hard can be much larger than the communication complexity of the problem.

We show that in the setting of one-way communication under product distributions, the dependence of communication cost on the allowed error $\epsilon$ is multiplicative in $\log(1/\epsilon)$ -- the previous upper bounds had the dependence of more than $1/\epsilon$.
This result explains how one-way communication complexity under product distributions is stronger than PAC-learning:\ both tasks are characterised by the VC-dimension, but have very different error dependence (learning from examples, it costs more to reduce the error).
\end{abstract}

\section{Introduction}

The standard way
 to attack the problem of showing a lower bound on the randomised communication complexity of a function $f$ is to choose a probability distribution $\mu$ on the inputs, and then show that the deterministic distributional complexity is large for $f$ w.r.t.\ $\mu$ -- i.e., that any deterministic protocol that computes $f$ with small error under $\mu$ must communicate much. This approach eliminates the need to argue about the randomness used by the protocol.\footnote{
We note that the popular information complexity method (see e.g.~\cite{bjks:datastream}) also uses distributional complexity, but does not seek to eliminate randomness from protocols.}

It is well known that this approach can be used without loss of generality, due to von Neumann's minimax theorem (see~\cite{kushilevitz&nisan:cc}; the same principle applies to many nonuniform computational models):
\[\max_\mu D_\epsilon^\mu(f)=R_\epsilon(f),\]
where $D_\epsilon^\mu(f)$ denotes the deterministic complexity of protocols that compute $f$ with error $\epsilon$ under the distribution $\mu$ of input to $f$, and $R_\epsilon(f)$ is the public coin randomised communication complexity of $f$ with worst-case error $\epsilon$.\footnote
{Throughout the paper we do not consider private coin randomised protocols.}

As a matter of convenience, one first tries to use a simple distribution $\mu$, for instance the uniform distribution, or more generally, product distributions over the inputs to Alice and Bob. This works for some problems, like Inner Product modulo 2 \cite{CG}. However, Babai, Frankl, and Simon \cite{bfs:classes} observed that for the Disjointness problem DISJ one cannot obtain lower bounds larger than $\Omega(\sqrt n\log n)$ under {\em any} product distribution, i.e., they show that an upper bound of $O(\sqrt n\log n)$ holds for every product distribution. They also give a lower bound of $\Omega(\sqrt n)$ under a product distribution. Later, Kalyanasundaram and Schnitger \cite{ks:disj} obtained the tight $\Theta(n)$ bound, and Razborov \cite{razborov:disj} showed that indeed $D^\mu_\epsilon(DISJ)=\Theta(n)$ for an explicit simple distribution $\mu$, for any sufficiently small constant $\epsilon>0$ (that such a $\mu$ exists is immediate from the result in \cite{ks:disj} and the minimax theorem, but their proof does not exhibit such a distribution explicitly). Distributional complexity under product distributions has been also frequently used to show structural properties like direct product theorems (e.g.,~\cite{jrs:comp,jkn:subdistribution}).
Furthermore, distributional communication complexity is the natural average case version of communication complexity, and it makes sense to study this for distributions that are `easy', in order to get a different model than randomised complexity. It seems natural to measure ``easiness'' via mutual information.

For many years it was open how large the gap between $R_\epsilon^{I=0}(f)=\max_{\mu\mbox{ product}}D_\epsilon^\mu(f)$ and $R_\epsilon(f)$ (for constant $\epsilon>0$) can be. Sherstov \cite{sherstov:prod} finally gave a proof that there are total Boolean functions $f$, where the former is $O(1)$ and the latter is $\Omega(n)$. In his result $f$ is not given explicitly. Recently Alon et al.~\cite{alon:vc} have given the following optimal explicit separation. Consider the problem where Alice gets a point and Bob a line from a projective plane containing $2^{\Theta(n)}$ points and lines. The VC-dimension of this problem is at most 2, which implies that the distributional complexity under any product distribution is $O(1)$ (even for one-way protocols, \cite{knr:rand1round}), whereas the sign-rank of the communication matrix is $2^{\Omega(n)}$, and hence the randomised (even unbounded error) communication complexity is $\Omega(n)$.

This leaves open a more precise investigation of the {\em amount} of correlation in $\mu$ needed to make $D^\mu(f)$ equal to $R(f)$.
It is natural to quantify this via the mutual information $I(X:Y)$, when the input $(X,Y)$ is drawn from $\mu$. We define the following measure:
\[R^{I\leq k}_\epsilon(f)=\max_{\mu: I(X:Y)\leq k}D^\mu_\epsilon(f).\]

We note here that the quantity on the right hand side does not change if randomised or deterministic protocols are allowed, because in the distributional setting
the randomness can be fixed without increasing the error (under any distribution).
The investigation of this measure has been initiated by Jain and Zhang \cite{JZ09} in the setting of one-way communication complexity (we discuss their contribution at the end of Section~\ref{ssec_dep}).
We note that $R^{I\leq n}_\epsilon(f)=R_\epsilon(f)$ for all functions $f:\{0,1\}^n\times\{0,1\}^n\to\{0,1\}$.

This family of complexity measures allows us to investigate how much correlation is needed in the input distribution to get good lower bounds.
We have 3 main applications.
First, we closely investigate the case of the Disjointness problem.
Second, we show that a certain problem exhibits a threshold behaviour, i.e., only with almost maximal correlation can a tight lower bound be proved, and this correlation can also be larger than the actual communication complexity of the problem.
Third, we investigate the dependence of one-way communication complexity under product distributions on the allowed error.

\subsection{The Disjointness problem}

In the \emph{Disjointness problem (DISJ)}, Alice and Bob receive, respectively, subsets $x,y\subseteq\{1,\ldots, n\}$, and their task is to decide whether $x$ and $y$ are disjoint.\fn
{We will often view $x$ and $y$ as binary \f n-bit strings, implicitly assuming the natural correspondence between the elements of $\01^n$ and $pow([n])$.}
This is one of the most-studied problem in communication complexity, which arguably has the biggest number of known applications to other models (see \cite{kushilevitz&nisan:cc}).
We give a complete characterisation of the information-bounded distributional complexity of Disjointness for all values of $k=I(X:Y)$, both in the randomised and in the quantum case.
\begin{theorem}\label{thm:cldisj}
For all $0\leq k\leq n$ and constant $\epsilon$ we have
\begin{enumerate}
\item  $R^{I\leq k}_\epsilon(DISJ)=\Theta(\sqrt{n(k+1)})$.
\item   $Q^{I\leq k}_\epsilon(DISJ)=\tilde{O}((n(k+1))^{1/4})$.
\item  $Q^{I\leq k}_\epsilon(DISJ)=\Omega((n(k+1))^{1/4})$.
\end{enumerate}
\end{theorem}

Previously, for classical protocols, a lower bound of $\Omega(\sqrt n)$ was known for a product distribution \cite{bfs:classes}, and the $\Omega(n)$ lower bound by Razborov \cite{razborov:disj} uses a distribution $\mu$ with $I^\mu(X:Y)=\Theta(n)$.
Babai et al.~\cite{bfs:classes} also gave an upper bound of $O(\sqrt n\log n)$ for product distributions, which we improve by a $\log$-factor. The quantum case has not been considered before.

Our results interpolate between the previously-known extreme cases, and also show that one needs input correlation $\Omega(n)$ to prove tight lower bounds.
Interestingly, the bounds depend inverse-polynomially on the error probability, except for the extreme cases of zero correlation and of maximal correlation.
We also note that a nearly-optimal complexity for randomised protocols can be achieved in a protocol with two rounds of communication (though not in one round).

The tight bound in the randomised case is based on a two-phase protocol, in which the players first remove ``uninteresting'' elements from their sets, until they are (essentially) small enough to be communicated.
For the quantum case this two-phase approach cannot be optimal, because the first phase reveals ``too much'' information about the input.
Therefore we give a completely different protocol for the quantum case, in which the players identify uninteresting elements a priori. This approach is tight up to a log-factor.

\subsection{Mutual information in hard distributions}

Note that for DISJ the complexity increases with the information parameter, and
the randomised communication complexity bound $\Theta(n)$ is reached only once the information in the hard distribution reaches $\Omega(n)$. For other problems like Inner Product mod 2 the tight bound of $\Omega(n)$ is reached already under product distributions \cite{CG}.
But can the mutual information between the input sides that is required to show a tight lower bound ever be {\em larger} than the actual communication complexity?
I.e., is it ever necessary to use distributions that are (much) more strongly correlated than the communication lower bound we want to show, or is it always possible to prove a tight lower bound for a (total) function $f$ by using a hard distribution with $I(X:Y)\leq poly(R(f))$?
A weak example is the quantum complexity of Disjointness, where the tight $\Omega(\sqrt n)$ bound is only reached when the information reaches $\Omega(n)$, but even here the complexity increases gradually with the information. We resolve this question, although our example is not explicit.

\begin{theorem} For every $n\leq d\leq 2^{n/100}$ there is a function $f_d:\01^n\times\01^n\to\01$ that has $R(f_d)=\Theta(\log d)$, but under all bipartite distributions with mutual information less than $n/1000$ the communication bound is $R^{I\leq n/1000}_{1/10}(f_d)\leq O(\log n)$.
\end{theorem}

Hence for $f_d$ the complexity stays low until the information is almost maximal, and then shoots up.

\subsection{Dependence of $R^{A\to B,I=0}_\epsilon(f)$ on $\epsilon$
\protect\footnote
{The same result has been obtained recently by Molinaro et al.~\cite{MWY:amow} independently.
The methods being used in the two works are similar;~\cite{MWY:amow} has been published prior to the current publication, while our results have been presented during a public talk at BIRS prior to either publication.}}
\label{ssec_dep}

Finally, we investigate the error dependence of $R^{I\leq k}_\epsilon(f)$ for arbitrary $f$.
In the unrestricted case, by standard boosting techniques we have $R_\epsilon(f)\leq O(R_{1/3}(f)\cdot\log(1/\epsilon))$. We call a function $f$ and a class $C$ of distributions on the inputs with
$\max_{\mu\in C} D^\mu_\epsilon(f)\leq O(\max_{\mu\in C}D^\mu_{1/3}(f)\cdot\log (1/\epsilon))$
{\em boost-able}. For this definition we require the above to be true for all $\epsilon$.
One can easily show that there are distributions $\mu$ and functions $f$, such that e.g.~$D^{\mu}_{1/4}((f)=\Omega(n)$ and $D^\mu_{1/3}(f)=0$, by placing a hard problem with weight $1/3$ in an otherwise constant matrix, so for a fixed distribution $\mu$ one cannot in general expect the error dependence to behave nicely.

Boost-ability is a property of a class of distributions.
The class of all distributions clearly has the property, but what about the class of distributions with information at most $I$?
In particular, what about $I=0$?

The issue is particularly interesting for product distributions, because boost-ability can be used to derive upper bounds on $R^{I\leq k}(f)$ from upper bounds on $R^{I=0}(f)$: due to the substate theorem (Fact \ref{fac:sub} below), a protocol that solves $f$ under all product distributions with error $\epsilon 2^{-9k/\epsilon}$ can be used to solve $f$ under distributions with $I(X:Y)=k$ with error $\epsilon$, hence boost-ability would imply $R^{I\leq k}_\epsilon(f)\leq O((k+1)\cdot R_{1/3}^{I=0}(f)/\epsilon)$ for all $f$.

We will use the super-script ``$A\to B$'' to denote one-way communication.
In this model the class of product distributions is boost-able:

\begin{theorem}\label{thm:learn}
$R^{A\to B,I=0}_\epsilon(f)\leq O(R^{A\to B,I=0}_{1/3}(f)\cdot \log(1/\epsilon))$.
\end{theorem}

We also show that when the information is between $1$ and $n^{1-\Omega(1)}$, then neither distributional randomised nor distributional quantum protocols are, in general, boost-able, see our Corollaries~\ref{coro:boost} and~\ref{coro:boostq}.

It is well known that $R^{A\to B,I=0}_{1/3}(f)=\Theta(VC(f))$ \cite{knr:rand1round}, where $VC(f))$ is the VC-dimension of the set of rows of the communication matrix. This even extends to the quantum case \cite{ambainis:racj,klauck:qpcom}.
The VC-dimension is also known to characterise the hardness of PAC-learning (see the monograph by Kearns and Vazirani \cite{kearns&vazirani:learn}) -- in fact, the previous proofs of the upper bound on $R^{A\to B,I=0}_\epsilon(f)$ in terms of VC-dimension have been done by explicitly simulating learning algorithms in the one-way communication model: Random examples are generated using a public coin, and Alice classified the examples in order to teach Bob a row of the communication matrix of $f$ in the PAC sense (examples were generated from the public coin, and Alice labelled those examples spending 1 bit per example).

The main limitation of this approach is that for PAC learning one needs $\Omega(1/\epsilon)$ examples to achieve error $1/\epsilon$. On the other hand, this approach ignores two strengths of the one-way model: First, Alice and Bob know the underlying distribution; second, Alice can do more than simply label examples.
One can interpret the one-way communication model under product distributions as a learning model, in which Alice is an (old-fashioned) teacher, who teaches by monologue, but using shared randomness that does not count towards the communication. Does such a teacher offer any advantage over learning from random examples? At first glance no, since both models are characterised by the VC-dimension, and one could conclude that learning from experience is all it takes.
Our Theorem~\ref{thm:learn}, however, shows that the final error can be made much smaller when learning from a teacher, comparing to learning ``just from experience''.
Note that in practice $1/\epsilon$ can also easily become the dominating factor in the complexity of a learning algorithm.

The main idea in our protocol is that Alice and Bob can beforehand agree on an $\epsilon$-net among the rows of the communication matrix, and Alice simply sends the name of the nearest row in the net. During a PAC learning algorithm, on the other hand, the $\epsilon$-net is generated from examples, which is more costly.

We can now discuss the previous result of Jain and Zhang \cite{JZ09}.
They show that for all total Boolean functions $f$ in the one-way model: \[R^{A\to B,I\leq k}_\epsilon(f)\leq O((k+1)\cdot R^{A\to B,I=0}_{1/3}(f)\cdot 1/\epsilon^2\cdot\log(1/\epsilon)).\]
This extends the VC-dimension upper bound to distributions with nonzero information.
Their protocol for information-$k$ distributions is constructed by simulating the PAC learning algorithm for the row $x$, and by generating
examples $y',f(x,y')$ using a rejection-sampling protocol.
We can improve the error dependence to $1/\epsilon$ by the following idea. Due to the Substate Theorem (Fact \ref{fac:sub} below) it is enough to find a protocol
that has error $2^{-9k/\epsilon}$ under the product of the marginal distributions of a distribution $\mu$ (with information $k$).
But this can be achieved with communication $O((k+1)/\epsilon\cdot R^{A\to B,I=0}_{1/3}(f))$ according to Theorem~\ref{thm:learn}.

\section{Preliminaries and Definitions}

\subsection{Information Theory}
We refer to \cite{cover&thomas:infoth} for standard definitions concerning information theory.

The relative entropy of two distributions on a discrete support is denoted by $D(\rho||\sigma)$. The relative max-entropy is $D_\infty(\rho||\sigma)=
\max_x\log(\rho(x)/\sigma(x))$. Note that these quantities are infinite, if the support of $\sigma$ does not contain the support of $\rho$. We mostly consider bipartite distributions on $\{0,1\}^n\times\{0,1\}^n$. The mutual information
is $I(X:Y)=D(\mu||\mu_X\times\mu_Y)$, where $\mu$ is the joint distribution of $(X,Y)$ and $\mu_X$, and $\mu_Y$ are the two marginal distributions of $\mu$.
We also use the quantity $I_\infty(X:Y)=D_\infty(\mu||\mu_X\times \mu_Y)$.
If we want to indicate the distribution used we write its name as a superscript, like $I^\mu(X:Y)$.
When $(X,Y,Z)\sim\mu$, we write $\mu(Y)$ to address the marginal distribution of $Y$, and sometimes (when we feel that a reader might benefit from such ``expansion'') we write $\mu(X,Y,Z)$ to address $\mu$ itself.

We first state the following well-known fact, see \cite{jrs:privacy}.
\begin{fact}[Substate Theorem]\label{fac:sub}
\begin{enumerate}
\item $I(X:Y)\leq I_\infty(X:Y)$.
\item For a given $\mu$ there is a $\mu'$ with $||\mu-\mu'||\leq\epsilon$, and $I_\infty^{\mu'}(X:Y)\leq I^\mu(X:Y)\cdot4/\epsilon$,
where $||\mu-\mu'||$ is the total variation distance between $\mu$ and $\mu'$.
\end{enumerate}
\end{fact}

We will use the following lemmas and facts. The first follows from the definition of relative entropy.

\begin{lemma}
Let $\mu$ be a bipartite distribution, $\rho=\mu_A\times\mu_B$, and $\sigma=\sigma_A\times \sigma_B$ any product distribution.

Then $D(\mu||\sigma)=D(\mu||\rho)+D(\rho||\sigma)=I^\mu(X:Y)+D(\rho||\sigma)$.
\end{lemma}

The following is a consequence of the log sum inequality.
\begin{lemma}\label{lem:lowerdisc}
Let $\mu,\sigma$ be distributions (for concreteness on $\{0,1\}^n\times \{0,1\}^n$), and $E$ an event. Then we have that $\sum_{x,y\in E} \mu(x,y)\log(\mu(x,y)/\sigma(x,y))\geq\max\{-1,\mu(E)\log(\mu(E)/\sigma(E))\}$.
\end{lemma}

\begin{lemma}\label{lem:rest}

Let $\mu,\sigma$ be distributions on $\{0,1\}^n\times\{0,1\}^n$, $E$ an event, and $\mu'$ the distribution $\mu$ restricted to $E$.
Furthermore, assume that under $\mu$ we have that $Prob(E)=\alpha$.
Then $D(\mu'||\sigma)\leq (D(\mu||\sigma)+1)/\alpha-\log\alpha$.
\end{lemma}

\begin{proof}
For all $x,y\in E$ we have $\mu'(x,y)=\mu(x,y)/\alpha$, otherwise $\mu'(x,y)=0$.
\begin{eqnarray*}
&&D(\mu||\sigma)\\
&=&\sum_{x,y} \mu(x,y)\log\left(\frac{\mu(x,y)}{\sigma(x,y)}\right)\\
&\stackrel{(*)}{\geq}&\sum_{x,y\in E} \mu(x,y)\log\left(\frac{\mu(x,y)}{\sigma(x,y)}\right)-1\\
&\geq&\sum_{x,y\in E}\mu'(x,y)\cdot\alpha\cdot\log\left(\frac{\mu'(x,y)\cdot\alpha}{\sigma(x,y)}\right)-1\\
&=&D(\mu'||\sigma)\cdot\alpha+\alpha\log\alpha-1,
\end{eqnarray*}
where for (*) we use Lemma \ref{lem:lowerdisc} with the event $\{0,1\}^n\times\{0,1\}^n-E$.
\end{proof}

We will use the following {\em rejection sampling} protocol from \cite{harsha:corr}.

\begin{fact}\label{fac:sample}
Let $\mu$ and $\nu$ be distributions on $\{0,1\}^n$ with $D(\mu||\nu)=k$. Assume that Alice and Bob both know $\nu$, and can create samples from $\nu$ using a public coin. Then Alice can send a message of expected length $k+2\log k+O(1)$ to Bob, which allows Bob (and Alice) to obtain a shared sample from the distribution $\mu$. The expectation is over the public coin tosses, and Bob's sample is distributed exactly with $\mu$.
\end{fact}

To bound the average information during a protocol we have the following lemma.

\begin{lemma}
Let $\mu$ be a distribution on $\{0,1\}^n\times\{0,1\}^n$, and $X,Y$ the corresponding random variables. Let $\cal R$ be a partition of $\{0,1\}^n\times\{0,1\}^n$ into rectangles $R$ (where $R$ also indicates the random variable induced under $\mu$). Then we have $I(X:Y)\geq I(X:Y|R)$.
\end{lemma}
\begin{proof}
Note that if $\cal R$ consists of only one rectangle covering everything, then we have $I(X:Y)=I(X:Y|R)$. We show that for any $\cal R$, if we refine $\cal R$ into $\cal R'$ by splitting a single rectangle $A\times B$ (w.l.o.g.~splitting the columns into $B_1$ and $B_2$), then $I(X:Y|R)$ does not increase. Since every $\cal R$ can be obtained by starting with a single rectangle and splitting rectangles, this implies the lemma.\footnote{This is obvious for partitions arising from deterministic communication protocols. The same is true without assuming a protocol. Instead of showing how to split rectangles to generate the partition we may argue that starting from a partition we can merge rectangles until only one rectangle is left. To show this it is enough to prove that in any partition of the matrix positions $U\times V$ into two or more rectangles there must be two rectangles that have the same row- or column-set. The case of two rectangles is trivial, and for partitions into $m\geq 3$ rectangles we first consider the case where a single rectangle spans all rows or all columns, in which the remaining part can be treated via induction. If there is no such rectangle, we can take any rectangle $R=A\times B$ from the partition, and consider the three regions $S=A\times (V-B), T=(U-A)\times B, Q=(U-A)\times (V-B)$. It is easy to see by induction that either $SQ$ or $TQ$ must contain two rectangles that share the same row- or column-sets.}

Denote by $\mu_R$ the distribution $\mu$ restricted to a rectangle $R=A\times B$ and re-scaled. Denote by $\nu_R$ the product of marginals of $\mu_R$.
We have that $\nu_R(x,y)=\mu(A,y)\cdot \mu(x,B)/\mu(R)^2$, where $\mu(A,y)=\sum_{x\in A}\mu(x,y)$ and $\mu(x,B)=\sum_{y\in B}\mu(x,y)$.

Since $I(X:Y|R)=\sum_R\sum_{x,y\in R}\mu(x,y)\log\frac{\mu(x,y)\cdot\mu(R)}{\mu(A,y)\cdot\mu(x,B)}$, when we split (a particular) $R=A\times B$ into
$R_1=A\times B_1$ and $R_2=A\times B_2$, the expression for $I(X:Y|R)$ is changed by adding

\begin{eqnarray*}
&&\sum_{x,y\in R_1}\mu(x,y)\log\frac{\mu(R_1)}{\mu(x,B_1)}+\sum_{x,y\in R_2}\mu(x,y)\log\frac{\mu(R_2)}{\mu(x,B_2)}\\
&=&\sum_{x\in A}\mu(x,B_1)\log\frac{\mu(R_1)}{\mu(x,B_1)}+\sum_{x\in A}\mu(x,B_2)\log\frac{\mu(R_2)}{\mu(x,B_2)},
\end{eqnarray*}
and subtracting
\[\sum_{x,y\in R}\mu(x,y)\log\frac{\mu(R_1)+\mu(R_2)}{\mu(x,B_1)+\mu(x,B_2)}=\sum_{x\in A}(\mu(x,B_1)+\mu(x,B_2))\log\frac{\mu(R_1)+\mu(R_2)}{\mu(x,B_1)+\mu(x,B_2)}.\]

The log sum inequality implies that the latter is not smaller than the former and we are done.
\end{proof}

The same follows for the transcript of a randomised protocol.

\begin{lemma}\label{lem:decrease}
Let $C$ be a random variable (the public coin) and ${\cal R}(c)$ be a partition of $\{0,1\}^n\times\{0,1\}^n$ into rectangles $R$ (depending on a value $c$ of $C$).
Let $\mu$ denote a distribution on $\{0,1\}^n\times \{0,1\}^n$, independent of $C$. Then $I(X:Y)\geq I(X:Y|R,C)$, where $R$ is the random variable that represents a rectangle from ${\cal R}_c$ when $C=c$.
\end{lemma}

\begin{proof}
$I(X:Y|R,C)=E_c I(X:Y|R,C=c)\leq I(X,Y|C=c)=I(X:Y)$.\end{proof}

The next lemma follows from a calculation and shows that a distribution can decrease a joint probability compared to the product of marginal distributions only in the presence of mutual information.

\begin{lemma}
Let $X,Y$ be Boolean random variables with a joint distribution $\mu$ and marginal distributions $\mu_A,\mu_B$. If $\mu_A(X=1)\mu_B(Y=1)\geq2\mu(X=Y=1)$, then
$I^\mu(X:Y)\geq \mu_A(X=1)\mu_B(Y=1)/5$.
\end{lemma}

Finally, we show that this is true for any product distribution, not just the product of marginals.

\begin{lemma}\label{lem:inf}
Let $X,Y$ be Boolean random variables with a joint distribution $\mu$ (and set $\rho=\mu_A\times\mu_B$), and $\sigma$ any product distribution. If $\sigma_A(X=1)\sigma_B(Y=1)
\geq 4\mu(X=Y=1)$ then $ D(\mu||\sigma)\geq\sigma(X=Y=1)/16$.
\end{lemma}

\begin{proof}
If $\rho(X=Y=1)\geq \sigma(X=Y=1)/2$, then by the above lemma $D(\mu||\sigma)\geq D(\mu||\rho)=I^\mu(X:Y)\geq \rho(X=Y=1)/5\geq\sigma(X=Y=1)/10$, because $\sigma$ is a product distribution and the relative entropy of $\mu$ and a product distribution is minimal for $\rho$.
If $\rho(X=Y=1)\leq \sigma(X=Y=1)/2$, then we can bound $D(\mu||\sigma)\geq D(\rho||\sigma)=D(\mu_A||\sigma_A)+D(\mu_B||\sigma_B)$. Assume that $\alpha=\mu_A(X=1)\leq\beta/\sqrt{2}=\sigma_A(X=1)/\sqrt{2}$. Then $(1-\alpha)\log((1-\alpha)/(1-\beta))+\alpha\log(\alpha/\beta)\geq\beta/16$. Hence in this case $D(\rho||\sigma)\geq D(\mu_A||\sigma_A)\geq\beta/16=\sigma_A(X=1)/16\geq\sigma_A(X=1)\sigma_B(Y=1)/16$. Other cases follow by symmetry.

\end{proof}

\subsection{Communication Complexity}

We assume familiarity with classical and quantum communication complexity. For the former consult \cite{kushilevitz&nisan:cc}, the latter is surveyed in \cite{deWolf:survey}. We concentrate on distributional complexity, which we define here.
\begin{definition}
The distributional complexity $D^\mu_\epsilon(f)$ is the minimal worst case communication cost of any deterministic protocol that computes $f$ with error $\epsilon$ under $\mu$. Similarly we define $R^\mu_\epsilon(f)$ for randomised public coin protocols and $Q^\mu_\epsilon(f)$ for quantum protocols (we consider quantum protocols with shared entanglement, but do not use the entanglement in our protocols).
When we drop the error $\epsilon$ from the notation, we set $\epsilon=1/3$.  When we drop the superscript we mean the ordinary, worst-case communication complexity.
\end{definition}

We observe that $R^\mu_\epsilon(f)=D^\mu_\epsilon(f)$ for all $f,\mu,\epsilon$, because one can fix the public coin randomness without increasing the error. Hence, we adopt the $R$-notation, and use randomness in upper bounds and deterministic protocols in lower bounds. Note that $Q^\mu_\epsilon(f)$ can  be smaller than $R^\mu_\epsilon(f)$, for instance for Disjointness under the hard distribution exhibited by Razborov \cite{razborov:disj}, where $R^\mu(DISJ)=\Theta(n)$, since the quantum complexity of DISJ is at most $O(\sqrt n)$ \cite{aaronson&ambainis:searchj}.

We consider functions $f:\{0,1\}^n\times\{0,1\}^n\to\{0,1\}$.
\begin{definition}
Define by $D(k)$ the set of distributions on the inputs that have $I(X:Y)\leq k$.

We define $R^{I\leq k}_\epsilon(f)=\max_{\mu\in D(k)} R^\mu_\epsilon(f)$ and use an analogous definition for the quantum case.
\end{definition}

Clearly $R(f)=R^{I\leq n}(f)$ and $R^{I=0}(f)$ is the complexity under the hardest product distribution.

\begin{definition}
One-way protocols allow only a single message from Alice to Bob, who produces the output. We indicate this model by a superscript, like $R^{A\to B, I\leq k}(f)$.
\end{definition}

Finally, we note the following fingerprinting technique \cite{kushilevitz&nisan:cc}.

\begin{fact}\label{fac:finger}
There is a public coin one-way protocol that checks equality of strings (of any length) with error $1/2^k$ and communication $k$.
\end{fact}

\section{Randomised Complexity of Disjointness}
\subsection{Upper Bound}

In this section we prove the upper bound for DISJ under bounded information distributions.

First we consider the case of 0 mutual information, for which we show an upper bound of $O(\sqrt n\log(1/\epsilon))$. Let $\mu$ be a product distribution on the inputs to DISJ. Babai et al.~\cite{bfs:classes} already show a protocol of cost $O(\sqrt n\log n\log(1/\epsilon))$ [they do not state the dependence on $\epsilon$, which is however easy to derive from their proof]. Note that one can combine their protocol for product distributions with the Substate Theorem (Fact \ref{fac:sub}) to get a bound of $O(\sqrt{n}(k+1)\log n/\epsilon)$ on the distributional complexity under distributions with information $k$: every distribution with information $k$ approximately sits with probability $1/2^{4k/\epsilon}$ inside the product of its marginal distributions, hence it is enough to use a product distribution protocol with very small error. This bound is worse in the dependence on $k$ than what is proved below.

\begin{theorem}\label{thm:prot1}
$R^{I= 0}_\epsilon(DISJ)\leq O(\sqrt{n}\cdot\log(1/\epsilon))$.
\end{theorem}

The proof is in the appendix. The main issue here is to achieve the small error dependence. The protocol has a 2-phase structure, where in phase 1, assuming that Bob holds a large set and that the probability that $x\cap y'=\emptyset$ is large, random $y'$ are drawn using the public coin and, if disjoint from $x$, removed from the universe (initially $\{1,\ldots, n\}$). After doing this sufficiently many times, the universe becomes small, and in phase 2 we use the small set disjointness protocol due to Hastad and Wigderson \cite{hastad:disj}.

Now we turn to distributions with more information. The protocol has the same structure, but we need to sample from a distribution of $y'$ that is not independent of $x$, which takes communication. The protocol also does not have the same error dependence, which we show is unavoidable later.
Due to this we may just analyse expected communication, and show that the worst case communication cannot be more than $1/\epsilon$ the established bound by appealing to the Markov bound.

\begin{theorem}\label{thm:prot2}
$R^{I\leq k}_\epsilon(DISJ)\leq O(\sqrt{n(k+1)}/\epsilon^2)$.
\end{theorem}

The proof is in the appendix. The main idea is to follow the 2-phase approach, and shrink the universe until is has size $S=\sqrt{n(k+1)}$. At this point the Hastad-Wigderson small set Disjointness protocol \cite{hastad:disj} takes over. To shrink the universe we need to sample inputs $y'$ from the distribution conditioned on $x$, and on being disjoint from $x$. This is achieved by using the rejection sampling protocol of Fact \ref{fac:sample}. We need to bound the information increase, but on average the first phase of the protocol removes $S$ elements from the universe using communication $O((k+1)/\epsilon)$. There are at most $n/S$ such iterations in phase 1, hence the expected communication is at most $O(n/S\cdot (k+1)/\epsilon)$. Another factor of $1/\epsilon$ is lost to turn this into a worst case bound by appealing to the Markov bound.

 In the next section we will also show a lower bound of $\Omega(\sqrt{n/\epsilon})$, so the error dependence cannot be made logarithmic, in contrast to the the 0 information case.

 One more issue we would like to consider is the number of rounds used. The above protocol can easily use a large number of rounds, and it is not immediately clear whether this is necessary. It is well known that the complexity of DISJ under product distributions for {\em one-way protocols} is $\Theta(n)$ \cite{knr:rand1round}.  We have the following modification that saves most of the interaction.

\begin{theorem}
\begin{enumerate}
\item The complexity of DISJ under distributions with information at most $k$ for protocols with 2 rounds is at most
$O(\sqrt{n(k+1)}\log n/\epsilon^2)$.
\item The complexity of DISJ under distributions with information at most $k$ for $O(\log^* n)$ rounds is at most $O(\sqrt{n(k+1)}/\epsilon^2)$.
\item In the case of 0 mutual information, the error dependence drops to a factor of $\log(1/\epsilon)$.
\end{enumerate}
\end{theorem}

\begin{proof}
For the first item we observe that in phase 1 Alice can simply act as if Bob's set was large, and continue to let him discover $y'_i$'s that are disjoint with $x$ until $U_i$ is guaranteed to be small. This does not increase the bound on the communication.
After this Bob can tell Alice, in which `round' his set really became small, so that she can recover the proper universe $U_j$. He also sends her his set using $\sqrt{n(k+1)}\log n$ bits. Note that in this protocol only Alice learns the result.

For the second item we do as above, but when Bob's set is small also repeat the same in reverse until both sets are small. Saglam and Tardos \cite{tardos:disj} have a protocol that solves small set disjointness in phase 2 in $O(\log^* n)$ rounds with communication $O(\sqrt{n(k+1)}\log(1/\epsilon))$.

Finally, note that for product distributions we can use the same modifications to the protocol described in Theorem \ref{thm:prot1}.
\end{proof}

\subsection{Lower Bound}

In this section we prove that the protocol of the previous section is optimal (except regarding the exact dependence on $\epsilon$).

For the lower bound we employ a distribution, depending on $n$ and $k$, such that the mutual information of the two inputs according to the marginal distributions is at most $k$; we then prove an $\Omega(\sqrt{n(k+1)})$ lower-bound for the distributional complexity under this distribution. In what follows we consider $k=k(n)$ as being $\in o(n)$, since for $k=\Omega(n)$ the upper bound on the information is trivial and the lower bound on the communication is known.

Let $c=\frac{1}{\log e}$ and $m=c\sqrt{n(k+1)}$. Note that $m=o(n)$ as well. Now $\mu_{n,k}$ can be defined as the distribution obtained by mixing two distributions. $\nu_{n,k}$ is uniform on pairs of disjoint subsets of $\{1,\ldots, n\}$ of size $m$, and $\sigma_{n,k}$ is uniform on pairs of subsets of size $m$ with an intersection of size $1$. Then $\mu_{n,k}=(3/4)\cdot\nu_{n,k}+(1/4)\cdot \sigma_{n,k}$.
This is essentially the distribution used in the proof by Razborov \cite{razborov:disj}, but with smaller sets.

We show in the appendix that the information is bounded by $k$.
\begin{theorem}\label{thm:low}
For any sufficiently small $\epsilon>0$ we have that $D^{\mu_{n,k}}_\epsilon(DISJ)=\Omega(\sqrt{n(k+1)})$, and hence that $R^{I\leq k}_{\epsilon}(DISJ)=\Omega(\sqrt{n(k+1)})$.
\end{theorem}

The proof is similar to that of the corresponding lower bound by Razborov \cite{razborov:disj} (for $k=\Theta(n)$).
A difficulty comes from the fact that Razborov's entropy ``counting'' argument no longer works in our case, because in that argument a linear number of terms have their entropy upper-bounded as $H\left(\frac{1}{2}\right)=1$. Since we still have to deal with a linear number of terms while having much less total entropy, we require a finer combinatorial counting argument instead.

Now we give a simple argument that shows that error dependence cannot be logarithmic in $1/\epsilon$.

\begin{theorem}\label{thm:error}
$R^{I\leq 1}_\epsilon(DISJ)=\Omega(\sqrt {n/\epsilon})$ for $\epsilon\geq\Omega(1/n)$.
\end{theorem}

\begin{proof}
Above we have described a distribution $\mu_{n,k}$ with information at most $k$ such that $\Omega(\sqrt{n(k+1)})$ communication is needed for some constant error $\delta$. We define $\tau_{n,k}$ to be $1/(2k)\cdot\mu_{n,k}+(1-1/(2k))\rho$, where $\rho$ is some product distribution for DISJ that puts weight 1/2 on 1-inputs.
Clearly, for error $\delta/{(4k)}$ the communication must be at least $\Omega(\sqrt{n(k+1)})$. Set $k=4\delta/\epsilon$ (note that $k\leq n$).

It remains to show that the information under $\tau$ is at most 1.
Let $E$ be an indicator random variable that indicates that $x,y$ have been chosen according to $\mu_{n,k}$. Then $I(X:Y)\leq I(XE:Y)=
I(E:Y)+I(X:Y|E)\leq H(E)+(1/2k)\cdot k\leq H(1/(2k))+1/2\leq 1$.

\end{proof}

\begin{corollary}\label{coro:boost}
The class of distributions with information $k$ with $1\leq k\leq n^{1-\Omega(1)}$ is not boost-able for randomised protocols.
\end{corollary}

\begin{proof}
Consider $k=1$. We have that $R^{I\leq 1}_{1/3}(DISJ)\leq O(\sqrt{n})$. If distributions with at most 1 bit information were boost-able, then
we would have $R^{I\leq 1}_\epsilon(DISJ)\leq O(\sqrt{n}\log(1/\epsilon))$. But the left hand side is at least $\Omega(\sqrt{n/\epsilon})$, which
puts a lower bound on $\epsilon$, whereas boost-ability should work for all $\epsilon$.

In the case of larger $k$ we use the same proof, to get that $\sqrt{\epsilon}\cdot\log(1/\epsilon)\geq\Omega(1/\sqrt{k+1})$, which remains a restriction on $\epsilon$ until
$k$ exceeds $ n^{1-\Omega(1)}$, and the assumption of Theorem \ref{thm:error} is violated.
\end{proof}

\section{Quantum Complexity of Disjointness}

\subsection{Upper Bound: First Attempt}

 Consider the two-phase approach from the previous section. The second phase `quantises' readily, if we do not care about log-factors: Simply use distributed quantum search by amplitude amplification to obtain a quadratic speedup in this part \cite{BuhrmanCleveWigderson98}. We mention here that the tight protocol for DISJ due to Aaronson and Ambainis \cite{aaronson&ambainis:searchj} does not seem to work well for the small set case, and so we do not know if the logarithmic factor is needed or not.

The problem is the first phase of the classical protocol, which seems impossible to quantise. Since phase 2 is now cheaper one can re-balance the costs of the two phases (details are left to the reader) and find a protocol with cost $\tilde{O}((n(k+1))^{1/3}$.

In the next section we will show that this bound is not optimal. We do note here, however, that the error dependence for the case $I(X:Y)=0$ is a factor of $O(-\log \epsilon)$ for the above, which will not be the case in the protocol we present next.

\ssect[ss_up]{Upper Bound: Almost Optimal Protocol}

We now describe a different approach that also works in the classical case, but loses a logarithmic factor and has worse error dependence for product distributions.
The approach we use identifies two moderately-sized blocks of ``interesting'' positions (i.e., the blocks are subsets of $[n]$), such that Alice can conduct a search efficiently on one block, and Bob on the other one.
The situation when the input sets intersect, but not on any interesting position, will be ``unlikely''.
This conforms to the intuition that if ``large'' $x$ and $y$ come from a product distribution that puts constant weight on 1-inputs, then there must be many ``semi-interesting'' or ``uninteresting'' positions -- i.e., such $i\in[n]$ that not both $i\in x$ and $i\in y$ is likely.

Let $\mu$ be a distribution on $\{0,1\}^n\times\{0,1\}^n$ with $I^\mu(X:Y)\leq k$. Denote by $E_i$ the event that $\sum_{1\leq j\leq i-1}X_jY_j=0$ -- i.e., $X$ and $Y$ are disjoint on $\{1,\ldots, i-1\}$.
Let $\mu_i$ be $\mu$, conditioned on $E_i$.
Let $q_i^x=Prob(Y_i=1|X=x, E_i)=Prob_{\mu_i}(Y_i=1|X=x)$ and $p_i^y=Prob(X_i=1|Y=y, E_i)=Prob_{\mu_i}(X_i=1|Y=y)$.\fn
{To the end of \sref{ss_up}, all the probabilities and the expectations are taken w.r.t.\ $\mu$, unless stated otherwise.}

{\bf The Protocol.}\\
Let $\tau=\dr{\eps^3}{\sq{(k+1)n}}$, $C_A=X\cap\set[q_i^X\ge\tau]i$ and $C_B=Y\cap\set[p_i^Y\ge\tau]i$.
Alice locally computes $C_A$ and, using Bob as an oracle, applies Grover's search to check whether $C_A\cap Y=\emptyset$ with error at most $\eps$ -- unless $\sz{C_A}\ge\fr{\log\dr1\eps}\tau$, in which case the protocol halts and declares ``$X\cap Y\neq\emptyset$''.
Bob does the same for $C_B$.
If an intersection has been found, the protocol declares ``$X\cap Y\neq\emptyset$''.
Otherwise, it is declared that ``$X\cap Y=\emptyset$''.

Intuitively, the protocol requires that each player searches among those positions, where he has `1' and the opponent is likely to have `1' as well.

The {\bf communication cost} of the protocol is
\m{\asO{\fr{\log^{\dr32}\fr1\eps\cdot\log n}{\sq\tau}}
 =\asO{\sqrt[4]{(k+1)n}\tm\log n\tm\lf(\fr{\log\dr1\eps}\eps\rt)^{\dr32}}.}

{\bf Error Analysis.}\\
The protocol can make a mistake in one of the following 3 cases: \bl{(a)} when $(C_A\cap Y)\cup(C_B\cap X)\neq\emptyset$ but this fact has not been detected\fn
{Note that $(C_A\cap Y)\cup(C_B\cap X) = \set[X_i=Y_i=1]{i\in C_A\cup C_B}$.};
\bl{(b)} when $\sz{C_A}\ge\dr1{\eps\tau}$ or $\sz{C_B}\ge\dr1{\eps\tau}$ but $X\cap Y=\emptyset$; \bl{(c)} when $(C_A\cap Y)\cup(C_B\cap X)=\emptyset$ but $X\cap Y\neq\emptyset$.
In the first case, the error can only happen if Grover's search fails -- with probability at most $\eps$ for each player (at most $2\eps$ in total); in the second case, the intersection was empty in spite of the fact that, conditioned on ``$X=x$'', the probability of this has been at most $(1-\tau)^{|C_A|}<\eps$, or similarly for ``$Y=y$'' and $|C_B|$ (or both) -- this adds less than $2\eps$ to the error.
So, the combined probability of the cases (a) and (b) is less than $4\eps$.

We will show that (c) occurs with probability \asO{\eps}.
For that we will introduce a classification of the positions $i\in[n]$, and we need a few more definitions to describe our classes.
By $\vec{x},\vec{y}$ we denote prefixes of strings $x,y$ of length $i-1$, where $i$ is usually clear from the context.
The random variable $\vec{X}$ is the prefix of length  $i-1$ of the random variable $X$ (Alice's inputs), and similarly for $Y$.

Denote $p_i^{\vec{x}}=Prob(X_i=1|E_i, \vec{X}=\vec{x})$, $q_i^{\vec{x}}=Prob(Y_i=1|E_i, \vec{X}=\vec{x})$, and similarly for $p_i^{\vec{y}}$ and $q_i^{\vec{y}}$.
Denote also $p_i'^{\vec{y}}=Prob(X_i=1|E_i,Y_i=1,\vec{Y}=\vec{y})$, and similarly for $p_i'^{\vec{x}}$, $q_i'^{\vec{x}}$ and $q_i'^{\vec{y}}$.
Denote by $r_i^{\vec{x}}=p_i^{\vec{x}}q_i'^{\vec{x}}=p_i'^{\vec{x}}q_i^{\vec{x}}$ the probability that $X_i=Y_i=1$ under the conditions $\vec{X}=\vec{x}$ and $E_i$, and similarly for other conditions (i.e., the super-script specifies the condition):\ say, $r_i$ is the probability that $X_i=Y_i=1$ conditioned on $E_i$, and so on.
Let $s_i=Prob(E_i)$, and like above, use super-scripts to specify conditions ($s_i^{\vec x}=Prob(E_i|\vec{X}=\vec{x})$, and so on).

We say that a position $i\in[n]$ is \e{bad for $x$} if $x_i=1$ and $q_i^{x}\leq\epsilon q_i'^{\vec{x}}$.
Informally, these are the positions where $x$ ``depresses'' the probability of intersection compared to $\vec x$.
Similarly define \pl[i] that are \e{bad for $y$} ($y_i=1$ and $p_i^y\leq\epsilon p_i'^{\vec y}$), and call $i$ \e{bad} if it is bad for $x$ \e{or} for $y$.

Say that $i$ is \e{lucky for $(\vec{x},\vec{y})$} if $p_i^{\vec{x}}\leq(k+1)p_i'^{\vec{y}}/\epsilon^3$.
Informally, these are the positions where the probability of the event ``$X_i=1$'', conditioned on $\vec{x}$ is not much higher than the probability of the same event, conditioned on $\vec{y}$ and (more importantly) on the event ``$Y_i=1$''.

Finally, we call every $i\in C_A\cup C_B$ \e{chosen}.
Recall that we are dealing with case (c) -- that is, we are analysing the probability that $X$ intersects $Y$ over the non-chosen coordinates.
There are 3 cases to consider:\ when the first intersection is on a bad position, when it is on
a non-chosen, non-bad, lucky position, and when it is on a non-chosen, non-lucky position.

Denote by $V_i$ the event that $i$ is bad for $x$, and by $W_i$ the event that $i$ is bad for $y$.
Set $\tilde{q}_i'^{\vec{x}}=Prob(Y_i=1 \wedge V_i|X_i=1,\vec{X}=\vec{x},E_i)$ and $\tilde{p}_i'^{\,\vec{y}}=Prob(X_i=1\wedge W_i|Y_i=1,\vec{Y}=\vec{y}, E_i)$.

For every $i\in[n]$, the probability that it is bad for $x$ and the first intersection is on $i$ is $p^{\vec{x}}_i \tilde{q}_i'^{\vec{x}} s_i^{\vec{x}}$.
Similarly, the probability that it is bad for $y$ and the first intersection is on it is $\tilde{p}_i'^{\vec{y}}q_i^{\vec{y}}s_i^{\vec{y}}$.

The following lemma shows that the first intersection cannot be on a bad position often.

\begin{lemma}\label{lem:tilde}
$\tilde{q}_i'^{\vec{x}}\leq\epsilon q_i'^{\vec{x}}$.
\end{lemma}

\begin{proof}
Denote by $V_i(x)$ the property that $x$ is A-bad for $i$ and by $E_i(x,y)$ the property that $x,y$ are disjoint on $\{1,\ldots,i-1\}$.
\begin{eqnarray*}
\tilde{q}_i'^{\vec{x}}&=&\frac{Prob(V_i\wedge Y_i=1\wedge X_i=1\wedge\vec{X}=\vec{x}\wedge E_i)}{Prob(X_i=1\wedge\vec{X}=\vec{x}\wedge E_i)}\\
&=&\sum_{x:x_i=1,x_1,\ldots, x_{i-1}=\vec{x},Bad(x,i)}\ \sum_{y:y_i=1, E_i(x,y)}\mu(x,y)/ Prob(X_i=1\wedge\vec{X}=\vec{x}\wedge E_i)\\
&=&\sum_{x:x_i=1,x_1,\ldots, x_{i-1}=\vec{x},Bad(x,i)}q_i'^{x} \cdot Prob(X=x\wedge E_i)/Prob(X_i=1\wedge\vec{X}=\vec{x}\wedge E_i)\\
&\leq&\eps q_i'^{\vec{x}}\cdot\sum_{x:x_i=1,x_1,\ldots x_{i-1}=\vec{x}} Prob(X=x\wedge E_i)/Prob(X_1=1\wedge\vec{X}=\vec{x}\wedge E_i)\\
 &\leq&\eps q_i'^{\vec{x}}.
\end{eqnarray*}
\end{proof}

Therefore,
\m{{\bf E}_x\sum_{i:V_i(x)}p_i^{\vec{x}}\tilde{q}_i'^{\vec{x}}s_i^{\vec{x}}
 \leq\eps\tm{\bf E}_x\sum_{i:V_i(x)}p_i^{\vec{x}}q_i'^{\vec{x}}s_i^{\vec{x}}
 \leq\eps\tm{\bf E}_x\sum_ip_i^{\vec{x}}q_i'^{\vec{x}}s_i^{\vec{x}}
 \leq\eps,
}

and similarly for $\tilde{p}_i'^{\bar{y}}q_i^{\vec{y}}s_i^{\vec{y}}$; so, the first intersection is on a bad position (either chosen or not) with probability at most $2\eps$.

Fix prefixes $(\vec{x},\vec{y})$. Assume $i$ is lucky for $(\vec{x},\vec{y})$.
Now consider some input $(x,y)$
consistent with $(\vec{x}, \vec{y})$, such that $i$ has not been chosen on $(x,y)$ and that $i$ is not bad for $(x,y)$.
If no such $(x,y)$ exists, then all non-chosen $i$ on all $(x,y)$ consistent with $(\vec x, \vec y)$ are bad, and the error on them has been accounted for above. Hence we assume we can find $(x,y)$ such that $i$ is not bad and not chosen.

We get
\m{p_i'^{\vec{y}} q_i'^{\vec{x}}
 \leq \fr{p_i'^{y}q'^{x}_i}{\eps^2}
 \leq \fr{\eps^4}{(k+1)n},}
where the first inequality holds because $i$ is not bad,
and the second one holds because $i$ is not chosen.

We still consider a lucky $i$ for $(\vec x,\vec y)$.
The probability, conditioned on $\vec{x}$, that the first intersection is at the position $i$ satisfies
\m{r_i^{\vec{x}}s_i^{\vec{x}}\leq r_i^{\vec{x}}
 = p_i^{\vec{x}}q_i'^{\vec{x}}
 \leq \fr{(k+1)\tm p_i'^{\vec{y}}q_i'^{\vec{x}}}{\eps^3}
 \leq\fr{\eps}n,}
where the second inequality holds because $i$ is lucky.
Therefore,
\[\sum_{i,\vec x,\vec y: \mbox{\small lucky}} Prob(\vec x,\vec y)r^{\vec x}_is^{\vec x}_i\leq\epsilon,\] where
we do not count error contributed by $i$'s that are bad on inputs $x,y$, but the first intersection of $x,y$.
So a first intersection occurs on a non-chosen lucky position with probability at most $\eps$ (again, ignoring error from bad $i,x,y$).

It remains to consider the case of non-chosen non-lucky positions.
The difference from the previous (lucky) case is that now
$p_i^{\vec{x}}>(k+1)p_i'^{\vec{y}}/\epsilon^3$ -- that is, being conditioned on $\vec{x}$, the event ``$X_i=1$'' is much more likely than conditioned on $\vec{y}$.
We will see that in this case an upper bound on the  probability of ``$X_i=Y_i=1$'' follows, essentially, from the assumption that $I^\mu(X:Y)\leq k$ -- the core assumption used in the proof of the following.

\begin{lemma}\label{lem:qinf}
Assume that for no $x$ or $y$ the conditional probability of non-intersection is less than $\alpha$ and that for no $x$, $y$ and $i$ the probability that $X_i=Y_i=1$ is larger than $1/2$ when conditioned on $E_i$, $E_i$ and $X=x$, or $E_i$ and $Y=y$ (i.e., $s_i\ge\alpha$ and $r_i^x,r^y_i,r_i\le1/2$ always).
Then
\m{\sum_i {\bf E}_{\vec x,\vec y}^{\mu_i}\,\ p_i^{\vec{x}}q_i^{\vec{y}}
 \leq 16 k/\alpha+68/\alpha^2.}
\end{lemma}

We will apply the lemma with $\alpha=\eps$.
Before we continue, let us see why we can assume that one-sided conditional probability of non-intersection is never less than $\eps$ and that $r_i^x$, $r^y_i$ and $r_i$ are never greater than $1/2$ for our target distribution.
If the original $\mu$ is not like that, we replace it with $\mu'$, obtained via sampling $(X,Y)\sim\mu$ and replacing $X$ by $0^n$ with probability $1/2$ and -- independently -- replacing $Y$ by $0^n$ with probability $1/2$.
Note that $\mu'$ satisfies the lemma requirement (as long as $\eps\le1/2$), $I^{\mu'}(X:Y)\leq I^\mu(X:Y)$ and any protocol solving $DISJ$ over $\mu'$ with error at most $\delta$ solves it over $\mu$ with error at most $4\delta$.

The probability that the first intersection is at a non-lucky non-chosen (and not necessarily non-bad) position satisfies
\m{\sum_i {\bf E}_{\vec y}^{\mu_i}\,p_i'^{\vec{y}}q_i^{\vec{y}}
 <\sum_i {\bf E}_{\vec x,\vec y}^{\mu_i}\,
  \fr{\eps^3p_i^{\vec{x}}q_i^{\vec{y}}}{k+1}
 \leq \fr{\eps^3\tm(16k/\eps+ 68/\eps^2)}{k+1}
 <84\eps,}
where the first inequality holds because $i$ is non-lucky and the second is \lemref{lem:qinf} with $\alpha=\eps$.

Taking into account all the steps of our analysis (including possible replacing of $\mu$ by $\mu'$), the error of our protocol is \asO\eps.
Via appropriate re-scaling of $\eps$, we get the required result:

\begin{theorem}
\f{Q^{I\leq k}_\epsilon(DISJ)\leq
 \asO{\sqrt[4]{(k+1)n}\tm\log n\tm\lf(\fr{\log\dr1\eps}\eps\rt)^{\dr32}}.}
\end{theorem}

It remains to prove the lemma.

\prfstart[\lemref{lem:qinf}]
The following argument extensively uses the assumption that
\m{k\geq I^\mu(X:Y)=D(\mu||\sigma),}
where $\sigma$ is the product of the marginals of $\mu$.
We need some new definitions:
For all $i,j\in[n]$, denote by $\sigma_i$ the product of the marginals of $\mu_i$, and by $\mu_i^{\vec{x},\vec{y},j}$ the distribution
$\mu_i$, conditioned on the event $X_1=x_1,\ldots, X_{j}=x_{j},Y_1=y_1,\ldots, Y_{j}=y_{j}$, which we abbreviate by $F^{\vec{x},\vec{y},j}$. Similarly, $\sigma^{\vec{x},\vec{y},j}_i$ is $\sigma_i$ conditioned on $F^{\vec{x},\vec{y},j}$. Note that for the latter probability distribution we first take the product of marginals of $\mu_i$, and then condition. This is different from considering conditional mutual information, in which one would first condition and then take the product of marginals.
We also stress that here $j$ denotes the length of $\vec x,\vec y$, unlike before.
In the following, when we do not mention $j$ explicitly, it is $i-1$:\ e.g., $\mu_i^{\vec x,\vec y}=\mu_i^{\vec x,\vec y,i-1}$.
Let
\m{k^{\vec x,\vec y}_i
 \deq D(\mu_i^{\vec{x},\vec{y}}(X_i,Y_i)||\sigma_i^{\vec{x},\vec{y}}(X_i,Y_i))}
and
\f{k_i \deq{\bf  E}_{\vec{x},\vec{y}}^{\mu_i}\, k^{\vec x,\vec y}_i.}

Observe that $p_i^{\vec{x}}q_i^{\vec{y}}$ is the probability that $X_i=Y_i=1$ under the distribution $\sigma_i^{\vec{x},\vec{y}}(X_i,Y_i)$.
As $\sigma_i$ is a product distribution, conditioning on $Y$ does note change the probability of $X_i=1$, and so,
\m{p_i^{\vec x}=Prob_{\sigma_i}(X_i=1|\vec{X}=\vec{x},\vec{Y}=\vec{y}).}
We can now use \lemref{lem:inf} to conclude that either $p^{\vec{x}}_iq^{\vec{y}}_i\leq 4 r^{\vec{x},\vec{y}}_i$
or $D(\mu^{\vec{x},\vec{y}}_i(X_i,Y_i)||\sigma^{\vec{x},\vec{y}}_i(X_i,Y_i))\geq p^{\vec{x}}_iq^{\vec{y}}_i/16$.
Hence,
\m{p_i^{\vec{x}}q_i^{\vec{y}}
 \leq 4r_i^{\vec{x},\vec{y}}
  + 16 D(\mu^{\vec{x},\vec{y}}_i(X_i,Y_i)||\sigma^{\vec{x},\vec{y}}_i(X_i,Y_i))
 = 4r_i^{\vec{x},\vec{y}} + 16k_i^{\vec x\vec y}}
and
\m{\sum_i {\bf E}^{\mu_i}_{\vec x,\vec y}\, p_i^{\vec x} q_i^{\vec{y}}
 \leq \sum_i {\bf E}^{\mu_i}_{\vec x,\vec y}
   \lf(4r_i^{\vec{x},\vec{y}} + 16k_i^{\vec x\vec y}\rt)
 =4\sum_i{\bf E}_{\vec x,\vec y}^{\mu_i} r_i^{\vec x,\vec y}
  +16\sum_i k_i.}
Note that
\m{\sum_i {\bf E}^{\mu_i}_{\vec x,\vec y} r_i^{\vec x,\vec y}s_i
 \leq \sum r_i s_i\leq 1,}
and since we are assuming that $s_i\ge\alpha$ always,
\m[m_Linf_1]{\sum_i {\bf E}_{\vec x,\vec y}^{\mu_i}\,\ p_i^{\vec{x}}q_i^{\vec{y}}
 \leq 4/\alpha + 16\sum_i k_i \leq 4/\alpha^2 + 16\sum_i k_i.}

It remains to bound $\sum k_i$ by $k/\alpha+4/\alpha^2$.
Note that for $(x,y)\in E_{i+1}$ it holds that $\mu_{i+1}(x,y)=\mu_i(x,y)/(1-r_i)$, and so,
\mal{\sigma_{i+1}(x,y) %
 &= \sum_{y':(x,y')\in E_{i+1}}\fr{\mu_i(x,y')}{1-r_i}
   \cdot\sum_{x':(x',y)\in E_{i+1}}\fr{\mu_i(x',y)}{1-r_i}\\
 &= \sum_{y':(x,y')\in E_i}\fr{\mu_i(x,y')\cdot (1-r_i^x)}{1-r_i}\cdot
   \sum_{x':(x',y)\in E_i}\fr{\mu_i(x',y)\cdot (1-r_i^y)}{1-r_i}\\
 &= \fr{\mu_i(x)\cdot\mu_i(y)\cdot (1-r_i^x)(1-r_i^y)}{(1-r_i)^2}\\
 &= \sigma_i(x,y)\dt\fr{(1-r_i^x)(1-r_i^y)}{(1-r_i)^2}.}
Then
\mal[P]{&D(\mu_{i+1}||\sigma_{i+1})\\
 =&\sum_{(x,y)\in E_{i+1}}\mu_{i+1}(x,y)
   \log\frac{\mu_{i+1}(x,y)}{\sigma_{i+1}(x,y)}\\
 =&\sum_{(x,y)\in E_{i+1}}\fr{\mu_{i}(x,y)}{1-r_i}
   \cdot\log\frac{\mu_{i}(x,y)\cdot(1-r_i)}{\sigma_{i}(x,y)(1-r_i^x)(1-r_i^y)}\\
\malabel{mal1}
 \leq&\sum_{(x,y)\in E_i}\fr{\mu_i(x,y)}{1-r_i}
   \cdot\log\frac{\mu_i(x,y)\cdot(1-r_i)}{\sigma_i(x,y)(1-r_i^x)(1-r_i^y)}\\
 &-\fr{\mu_i(E_i\smin E_{i+1})}{1-r_i}
   \cdot\fr{\log(4\mu_i(E_i\smin E_{i+1})}{\sigma_i(E_i\smin E_{i+1})}\\
 \leq& \sum_{(x,y)\in E_i}\fr{\mu_i(x,y)}{1-r_i}
   \cdot\log\frac{\mu_i(x,y)}{\sigma_i(x,y)}\\
\malabel{mal15}
 &+\sum_{(x,y)\in E_{i}}\fr{\mu_{i}(x,y)}{1-r_i}
   \cdot\log\frac{1-r_i}{(1-r_i^x)(1-r^y_i)}-\fr{r_i\cdot\log(4r_i)}{1-r_i}\\
\malabel{mal2}
 \leq& \fr{D(\mu_i||\sigma_i)}{1-r_i} + 2\sum_{(x,y)\in E_{i}}\mu_{i}(x,y)
   \cdot 2\cdot (r^x_i+r^y_i) -2r_i\log(4r_i)\\
 \leq& \fr{D(\mu_i||\sigma_i)}{1-r_i} + 12 r_i-2r_i\log(r_i),}
where in \bref{mal1} we used Lemma \ref{lem:lowerdisc},
in \bref{mal15} substituted $r_i=\mu_i(E_i\smin E_{i+1})$,
in \bref{mal2} used the fact that $-\log(1-\lambda)\leq 2\lambda$ for $0\leq\lambda\leq 1/2$, and repeatedly applied $r_i^x,r^y_i,r_i\leq1/2$.
The above inequality tells us that the corresponding relative entropy increases only slightly between the positions $i$ and $i+1$.

Next we fix $X_1=x_1,\ldots, X_{i}=x_i$ and $Y_1=y_1,\ldots,Y_i=y_i$, and look at the term $D(\mu^{\vec{x},\vec{y},i}_{i+1}||\sigma^{\vec{x},\vec{y},i}_{i+1})$, comparing it to $D(\mu^{\vec{x},\vec{y},i}_i||\sigma^{\vec{x},\vec{y},i}_i)$.
Recall that the involved distributions $\mu^{\vec{x},\vec{y},i}_{\dots}$ and $\sigma^{\vec{x},\vec{y},i}_{\dots}$ are on $X_{i+1},\ldots, X_n,Y_{i+1},\ldots, Y_n$, and they are \e{determined} by the values $\vec{x}=x_1\dc x_i$ and $\vec{y}=y_1\dc y_i$.
Below we assume that $x_jy_j\neq 1$ for all $j\le i$ (otherwise the input is not in the support of $\mu_{i+1}$ and the following upper bound on $D(\mu_{i+1}^{\vec x,\vec y,i}||\sigma_{i+1}^{\vec x,\vec y,i})$ holds trivially).
\mal[P]{
&D(\mu_{i+1}^{\vec x,\vec y,i}||\sigma_{i+1}^{\vec x,\vec y,i})\\
\malabel{mal25}
=&\sum_{x,y\in E_{i+1}: (x_1,\ldots, x_i)=\vec x, (y_1,\ldots, y_i)=\vec y} \mu_{i+1}(x,y|\vec x,\vec y)\log\left(\frac{\mu_{i+1}(x,y|\vec x,\vec y)}{\sigma_{i+1}(x,y|\vec x,\vec y)}\right)\\
\malabel{mal3}
\leq&\sum_{x,y\in E_{i}: (x_1,\ldots, x_i)=\vec x, (y_1,\ldots, y_i)=\vec y}\mu_i(x,y|\vec x,\vec y)\log\left(\frac{\mu_i(x,y|\vec x,\vec y)}
{\sigma_i(x,y|\vec x,\vec y)\cdot (1-r^{x}_i)(1-r^y_i)}\right)\\
\leq& D(\mu_i^{\vec x,\vec y,i}||\sigma_i^{\vec x,\vec y,i})
 +\sum_{x,y: (x_1,\ldots, x_i)=\vec x, (y_1,\ldots, y_i)=\vec y} \mu_i(x,y|\vec x,\vec y)\cdot\log\left(\frac{1}{(1-r^x_i)(1-r^y_i)}\right)\\
\leq& D(\mu_i^{\vec x,\vec y,i}||\sigma_i^{\vec x,\vec y,i})
 +\sum_{x,y: (x_1,\ldots, x_i)=\vec x, (y_1,\ldots, y_i)=\vec y}
   \mu_i(x,y|\vec x,\vec y)(2r^x_i +2r^y_i)\\
=& D(\mu_i^{\vec{x},\vec{y},i}||\sigma_i^{\vec{x},\vec{y},i})+2 r_{i}^{\vec y}+2r_i^{\vec x},}
where in \bref{mal25} and \bref{mal3} the condition $E_{i+1}$ is satisfied by all the considered inputs (this is implied by the values $(\vec x,\vec y)$), and so, no ``re-scaling'' happens while going from $\mu_{i+1}$ to $\mu_i$.
In particular,
\m[m_part]{{\bf E}^{\mu_i}_{\vec x,\vec y}\,
  D(\mu_{i+1}^{\vec x,\vec y,i}||\sigma_{i+1}^{\vec x,\vec y,i})
 \le{\bf E}^{\mu_i}_{\vec x,\vec y}\,
  D(\mu_i^{\vec{x},\vec{y},i}||\sigma_i^{\vec{x},\vec{y},i}) +4r_i.}

We are ready to bound $\sum_i k_i=\sum_i {\bf E}^{\mu_i}_{\vec x,\vec y} \ k^{\vec x,\vec y}_i$.
Note that $\mu_1=\mu$, let us use the ``chain rule'' for relative entropy:
\mal[P]{
k\geq&D(\mu||\sigma)\\
=& D(\mu_1(X_1,Y_1)||\sigma_1(X_1,Y_1))\\
&+{\bf E}^{\mu_1}_{x_1,y_1} D(\mu_1^{x_1,y_1}(X_2,\ldots, X_n,Y_2,\ldots, Y_n)||\sigma_1^{x_1,y_1}(X_2,\ldots, X_n,Y_2,\ldots, Y_n))\\
\malabel{mal4}
\geq & D(\mu_1(X_1,Y_1)||\sigma_1(X_1,Y_1))
 +{\bf E}^{\mu_1}_{x_1,y_1}
   D(\mu_2^{x_1,y_1}||\sigma_2^{x_1,y_1})-4r_1^{}\\
\geq & D(\mu_1(X_1,Y_1)||\sigma_1(X_1,Y_1))
 +{\bf E}^{\mu_2}_{x_1,y_1}
   D(\mu_2^{x_1,y_1}||\sigma^{x_1,y_1}_2)\cdot(1-r_1)-4r_1^{}\\
=& D(\mu_1(X_1,Y_1)||\sigma(X_1,Y_1))\\
&+ {\bf E}^{\mu_2}_{x_1,y_1} D(\mu^{x_1,y_1}_2(X_2,Y_2)||\sigma^{x_1,y_1}(X_2,Y_2))\cdot(1-r_1)\\
&+ {\bf E}^{\mu_2}_{x_1,x_2,y_1,y_2} D(\mu^{x_1,x_2,y_1,y_2}_2(X_3,\ldots)||\sigma_2^{x_1,x_2,y_1,y_2}(X_3,\ldots))\cdot(1-r_1)-4r_1\\
\malabel{mal5}
\geq& D(\mu_1(X_1,Y_1)||\sigma_1(X_1,Y_1))+{\bf E}^{\mu_2}_{x_1,y_1}
  D(\mu^{x_1,y_1}_2(X_2,Y_2)||\sigma^{x_1,y_1}_2(X_2,Y_2))\cdot(1-r_1)\\
&+ {\bf E}^{\mu_2}_{x_1,x_2,y_1,y_2}
  D(\mu^{x_1,x_2,y_1,y_2}_3||\sigma_3^{x_1,x_2,y_1,y_2})\cdot(1-r_1)
   -4r_1-4r_2\cdot(1-r_1)\\
\geq& D(\mu_1(X_1,Y_1)||\sigma_1(X_1,Y_1))+{\bf E}^{\mu_2}_{x_1,y_1}
  D(\mu^{x_1,y_1}_2(X_2,Y_2)||\sigma^{x_1,y_1}_2(X_2,Y_2))\cdot(1-r_1)\\
&+ {\bf E}^{\mu_3}_{x_1,x_2,y_1,y_2}
  D(\mu^{x_1,x_2,y_1,y_2}_3||\sigma_3^{x_1,x_2,y_1,y_2})\cdot(1-r_1)(1-r_2)
   -4r_1-4r_2\cdot(1-r_1)\\
\vdots\\
\geq&\sum_i {\bf E}^{\mu_{i}}_{\vec x,\vec y}D(\mu^{\vec x,\vec y}_i(X_i,Y_i)||\sigma^{\vec x,\vec y}_i(X_i,Y_i))\cdot\alpha-4\sum_i r_i\\
=&\sum_i k_i\cdot\alpha-4/\alpha,}
where \bref{mal4} and \bref{mal5} followed from \bref{m_part} and in the last inequality we used the assumption that $\prod_{i=1}^n(1-r_i)\geq\alpha$ and
\m{\sum_{i=1}^nr_i\le\prod_i(1+r_i)\le\prod_i\fr1{1-r_i}\le\fr1\alpha.}

This means that $\sum k_i\leq k/\alpha+4/\alpha^2$, as required.
\prfend  

\subsection{Lower Bound}

We use exactly the same hard distribution for the quantum case as for the classical case, see Section 3.2, where also the mutual information of this distribution is shown to be at most $k$. Conveniently, Razborov \cite{razborov:qdisj} has done most of the hard work for us by analysing the quantum complexity of Disjointness for all set sizes. We get the following main result:

\begin{theorem}
The distributional quantum communication complexity of Disjointness under $\mu_{n,k}$ is at least $\Omega((n(k+1))^{1/4})$.
\end{theorem}

\begin{proof}

Recall the distributions $\nu_{n,k},\sigma_{n,k}$ as defined in Section 3.2. These are the distributions of sets of size $s=O(\sqrt{n(k+1)}$ from a size $n$ universe (not intersecting resp.~intersecting).
We employ the following result by Razborov \cite{razborov:qdisj}:

\begin{fact}\label{fac:razb2}
Any quantum protocol that solves DISJ with error $\epsilon$ under $\nu_{n,k}$ and error $\epsilon$ under $\sigma_{n,k}$ needs communication $\Omega(\sqrt s)=\Omega(({n(k+1)})^{1/4})$.
\end{fact}

This follows from Razborov's proof, in which given a quantum protocol with communication $c$ for DISJ (on inputs of size $s$ from a size $n$ universe), a uni-variate polynomial of degree $O(c)$ on $\{0,1,\ldots, s\}$ is constructed such that $p(i)$ is close to 0 for all $\{0,1,\ldots, s-1\}$ and $p(s)=1$. Such a polynomial must have degree $\Omega(\sqrt s)$.
The construction is done by averaging of the acceptance probabilities on all inputs $x,y$ where $x,y$ have size $s$, and hence it is enough if the given protocol for DISJ is correct on average inputs under $\nu_{n,k}$ and under $\sigma_{n,k}$. But any protocol with small error under $\mu_{n,k}$ must also have small error under both of these distributions, and we get the same lower bound under this distribution as in the worst case, as stated by Razborov.
\end{proof}

We also note that again, the error dependence cannot be poly-logarithmic.
The proof is the same as in the classical case.
\begin{theorem}
$Q^{I\leq 1}_\epsilon(DISJ)\geq\Omega((n/\epsilon)^{1/4})$ for $\epsilon\geq\Omega(1/n)$.
\end{theorem}

We again obtain this following.

\begin{corollary}\label{coro:boostq}
The class of distributions with information $k$ with $1\leq k\leq n^{1-\Omega(1)}$ is not boost-able for quantum protocols.
\end{corollary}

\section{Large Correlation is Needed for Tight Bounds}

In this section we show that there is a function, for which the distributional communication complexity is far from the randomised communication complexity if the information in the distribution is less than $\Omega(n)$. The main idea is that random sparse problems make it hard for low information distributions to `focus' on the 1-inputs.

Define $f_{n,d}$ as a random variable that takes as its values functions $f:\{0,1\}^n\times\{0,1\}^n\to\{0,1\}$. The functions are generated randomly as follows.
Each input $x,y$ is chosen to be a 1-input independently with probability $d/2^n$.

Note that the communication matrix of $f_{n,d}$ has expected $d$ 1-inputs for each row and column. In the following $d$ should be thought of as some value like $2^{\sqrt n}$.
We need $2^{n/100}\geq d\geq 6n$.

We first show that the complexity of $f_{n,d}$ is $\Theta(\log d)$ with high probability.
Then, we show that with high probability $f_{n,d}$ has a property that allows an $O(\log n)$ protocol under all low information distributions.

First we note that by the Chernoff bound the probability that a row or column has more than $2d$ or less than $d/2$ 1-inputs is at most $2e^{-d/3}\leq 2^{-2n}$. By the union bound it is true for all rows and columns (with high probability) that they contains between $d/2$ and $2d$ 1-inputs. Throughout this section we assume that $f_{n,d}$ has this property.

\begin{lemma} $R(f_{n,d})\leq O(\log d)$ with high probability.
\end{lemma}

\begin{proof}
With high probability there are at most $2d$ 1-inputs $(x_1,y),\ldots, (x_{2d},y)$ in Bob's column. If Alice sends a fingerprint of $x$ as in Fact \ref{fac:finger}, using $2\log d$ bits, then Bob can check whether $x=x_j$ for some $1\leq j\leq 2d$ with error $2d\cdot2^{-2\log d}\leq 2/d$. If so, then he accepts, otherwise he rejects.
\end{proof}

\begin{lemma}
$R(f_{n,d})\geq\Omega(\log d)$ with high probability.
\end{lemma}

\begin{proof}

The proof is by the probabilistic method.
We use the minimax theorem and the following hard distribution: Put 1/2 weight on 1-inputs and 1/2 weight on 0-inputs to $f_{n,d}$.
Note that the mutual information of this distribution is $\Omega(n)$: for 1-inputs, given $x$ there are at most $d$ inputs $y$ out of $2^n$ such that $x,y$ is a 1-input.
Hence the information is at least $(n-\log d)/2$.

We employ a 1-sided version of the discrepancy method (this is a relaxation of the 1-sided rectangle corruption method). The 1-sided discrepancy under a distribution $\mu$ is $disc'(f,\mu)=\max_R \mu(f^{-1}(1)\cap R)-\mu(f^{-1}(0)\cap R)$, where the maximum is over all rectangles. Then $R^\mu(f)\geq -\log disc'(f,\mu)-O(1)$ for all $\mu$ that put weight 1/2 on the 1-inputs. Our goal is to show that the 1-sided discrepancy is small with high probability over the choice of $f_{n,d}$.

Fix a rectangle $R$ and consider a random $f_{n,d}$.
We would like to compute the probability that $disc'(R)=\mu(f_{n,d}^{-1}(1)\cap R)-\mu(f_{n,d}^{-1}(0)\cap R)$ is large. Note that this is a random variable and that $\mu$ depends on $f_{n,d}$

If $\mu(R\cap f_{n,d}^{-1}(1))\leq 4/d^{1/4}$, then $disc'(R)\leq 4/d^{1/4}$ and we are done. Hence we assume the opposite.
For $R$ to contain at least a $4/d^{1/4}$ fraction of all 1-inputs it must be the case that $R$ contains at least $(4/d^{1/4})\cdot 2^n d/2$ 1-inputs, and no row or column contains more than $2d$ of them, which implies that $R$ must have at least $2^n/d^{1/4}$ rows and columns.

Write $R=A\times B$, where $|A|,|B|\geq 2^n/d^{1/4}$. The expected number of 1-inputs in $R$ is at most $|A|\cdot|B|\cdot d/2^n$. The 1-inputs are chosen independently, and the Chernoff bound yields that $Prob(R$ contains more than $(1+d^{-1/2})|A||B|d/2^n$ 1-inputs$)\leq e^{-|A||B| d/(3\cdot 2^n d^{1/4})}\leq e^{-2^n d^{1/4}/3} $.
Similarly, we can bound $Prob(R$ contains less than $(1-d^{-1/2})|A||B|d/2^n$ 1-inputs$)$.

Furthermore, since there are at most $2^{2^{n+1}}$ rectangles, by the union bound with high probability these estimates are correct for {\em all} rectangles with enough rows and columns (in particular the rectangle consisting of all inputs).

Note that $R$ contains at least $|A|\cdot|B|-|A|2d$ 0-inputs, each of which have weight at least $1/(2^{2n+1})$, for a total 0-weight of
at least $|A||B|/2^{2n+1}-d/2^{n}$. The weight of a single 1-input is at most $1/(1-d^{-1/2})\cdot 1/(d2^{n+1})$ and the total 1-weight of $R$ is at
 most $(1+d^{-1/2})/(1-d^{-1/2})\cdot|A||B|/2^{2n+1}$ by the above. Hence the one-sided discrepancy is at most $O(d^{-1/2}|A||B|/2^{2n+1})\leq O(d^{-1/2})$.
\end{proof}

We will now show that most functions $f_{n,d}$ are easy under all low information distributions, but hard for information $n$ distributions, by showing that $f_{n,d}$ has a certain property with high probability. We assume in the following that $d\leq2^{\epsilon^2 n}$ and set $\epsilon=1/10$.

\begin{definition}
We say a Boolean $2^n\times 2^n$ matrix is {\em good}, if it is true that every rectangle  $A\times B$ with $\min\{|A|,|B|\}\leq 2^{2n/3}$ has no more than $100\max\{|A|,|B|\}$ 1-entries. We also call any rectangle $A\times B$ with $\min\{|A|,|B|\}\leq 2^{2n/3}$ in a good matrix {\em good}.
\end{definition}

\begin{lemma}
With high probability the communication matrix of $f_{n,d}$ is good.
\end{lemma}

\begin{proof}
Fix $A,B$. Assume that $|B|\geq |A|$ and that $|A|\leq 2^{2n/3}$. The probability that a fixed $x,y$ is a 1-input is $d/2^n$. The probability that there are at least $100|B|$ 1-inputs in $R$ is at most
${|A||B|\choose 100|B|}\cdot (d/2^n)^{100|B|}\leq (\frac{|A|d}{2^n})^{100|B|}\leq\frac{d}{2^{n/3}}^{100|B|}$.

There are ${2^n\choose|A|}{2^n\choose |B|}\leq (e2^n/|B|)^{2|B|}$ rectangles of this size. By the union bound the probability that there is a rectangle that is not good is small.
\end{proof}

Now assume that $f$ (or rather its matrix) is good. Consider any $\nu$ such that $I(X:Y)\leq\epsilon^3 n$. We have to give a protocol for $f$ under $\nu$.
By Fact \ref{fac:sub} there is another distribution $\mu$, that is $\epsilon/2$-close to $\nu$ and has $I_\infty(X:Y)\leq 8\epsilon^2n$. We describe a protocol for $f$ under $\mu$ with error $\epsilon/2$. The same protocol has error at most $\epsilon$ under $\nu$.
We assume $d\leq 2^{\epsilon^2n}$.

Alice and Bob consider the marginal distributions $\mu_A$ and $\mu_B$.
Alice sends 0, if $\mu_A(x)\leq2^{-n/2-\epsilon n}$, and 1 otherwise, while Bob does the same for $\mu_B(y)$.
We first consider the rectangle $R_{00}$ of inputs on which the messages were $00$.
Then $\mu_A(x)\cdot\mu_B(y)\leq2^{-n-2\epsilon n}$ for all $x,y$ in $R_{00}$. Hence on this rectangle
$\sum_{x,y\in R:f(x,y)=1}\mu_A(x)\mu_B(y)\leq 2d2^{-2\epsilon n}$. That means that under $\mu_A\times \mu_B$ the probability of 1-inputs in $R_{00}$ is at most $2d2^{-2\epsilon n}$. But since $I_\infty(X:Y)\leq 8\epsilon^2n$, the probability of 1-inputs there under $\mu$ is at most $2^{-2\epsilon n+O(\epsilon^2n)}$. We can reject on $R_{00}$ without introducing much error.

Now consider one of the remaining rectangles, say $R_{10}=A\times B$ (the rectangle where Alice sent 1 and Bob 0).
Clearly $|A|\leq 2^{n/2+\epsilon n}$. Assume $|A|\leq |B|$.
By the above lemma this means that $A\times B$ is good, i.e., contains relatively few 1-inputs, on average only $100$ per column.

On $R_{10}$ Alice and Bob can send public coin fingerprints of $x,y$ each, with error guarantee $\epsilon/1000$ (see Fact \ref{fac:finger}). This takes communication $O(-\log\epsilon)$. If a column (or row) contains few 1-inputs Alice resp.~Bob can test with the fingerprint whether $x,y$ is one of these. But $R_{10}$ only contains few 1-inputs only on average, and it is quite possible that both the row and the column of $x,y$ have many 1-inputs.
Namely, while the (uniformly) average column has at most 100 1-inputs, the distribution on $R_{10}$ is not uniform, but instead rows with more 1-inputs have increased probability, so we need to be more careful.

Let $A=A_0$ and $B=B_0$ (we still assume that $|A|\leq |B|)$.
Define $A_i$ as the set of $x\in A_{i-1}$ such that there are at least 1000 1-inputs $x,y'$ with $y'\in B_{i-1}$ and $B_i$ the set of $y\in B_{i-1}$ such that there are at least 1000 1-inputs $x',y$ with $x'\in A_{i-1}$.

Clearly, all $A_i\times B_i$ are good.
Assume that $|A_i|\leq |B_i|$. $A_i\times B_i$ has at most $100 |B_i|$ 1-inputs.
$A_i\times B_{i+1}$ has at least $1000|B_{i+1}|$ 1-inputs, hence $|B_{i+1}|\leq |A_i|/10$, because $A_i\times B_{i+1}$ is good: $1000|B_{i+1}|\leq100\max\{|A_i|,|B_{i+1}|\}$. That means that for odd $i$ we have $|B_{i}|\leq |A_{i-1}|/10$ and for even $i$ we have $|A_i|\leq |B_{i-1}|/10$.

All sets $A_i$, $B_i$ are known to Alice and Bob without communication. Also, due to the shrinking sizes, all $i\leq O(n)$.

The protocol works as follows: Alice determines the first $i$ such that on $A_i\times B_{i-1}$ her row contains at most 1000 1-inputs and sends this information.
Bob also sends the index $j$, such that on $A_{j-1}\times  B_j$ his column contains at most 1000 1-inputs. If $i<j$, then Bob also sends a fingerprint of $y$ with error guarantee $1/10000$ (see Fact \ref{fac:finger}). If there is a $y'\in B_{i-1}$ with the same fingerprint and $f(x,y')=1$ then Alice accepts, otherwise she rejects. If $i>j$, then Alice sends the fingerprint, and Bob accepts if and only if there is an $x'\in A_{j-1}$ with $f(x',y)=1$.
Clearly the communication is $2\log n+O(1)$, and is done in 2 rounds.

Correctness: Assume $i<j$. The players can be sure that $x,y\in A_i\times B_{i-1}$. There are at most $1000$ 1-inputs in row $x$ in $B_{i-1}$. If $f(x,y)=1$, then certainly the fingerprints will coincide, and Alice accepts.
Otherwise the probability that the fingerprints equal is at most $100/10000=1/10$.

\begin{lemma} Under any  $\nu$ with information at most $\epsilon^3 n$ and for $6n\leq d\leq 2^{\epsilon^2n}$ we have that $R^\nu_\epsilon(f_{n,d}) \leq O(\log n)$, if $f_{n,d}$ is good.
\end{lemma}

\begin{theorem}
For every $6n\leq d\leq 2^{n/100}$ there is a function $f_{n,d}$ such that

\begin{itemize}
\item $R(f_{n,d})=\Theta(\log d)$,
\item $R^{I\leq n/1000}_{1/10}(f_{n,d})\leq O(\log n)$.
\end{itemize}

\end{theorem}

\section{One-Round Error Dependence}

We now consider the general question of error dependence under distributions with limited information. In the case, where the information is bounded only by $n$, we get the standard randomised (resp.~quantum) communication complexity, for which the usual boosting techniques
(i.e., the Chernoff bound) show that the error dependence is at most factor of $O(\log (1/\epsilon))$. Furthermore, Corollary \ref{coro:boost} shows that for all information parameters $1<k<n^{1-\Omega(1)}$
the error dependence is polynomial. This leaves the case of product distributions, where in the randomised two-way communication case
DISJ has logarithmic error dependence. In this section we show that for {\em all} total functions, in the case of one-way communication
complexity the error dependence is small under product distributions. The corresponding statement about two-way protocols remains open.

In \cite{knr:rand1round} Kremer et al. show that the complexity of one-way protocols for total functions under product distributions is determined by the VC-dimension (see also \cite{kearns&vazirani:learn}).

\begin{definition}
The VC-dimension of a Boolean matrix $M$ is the largest $k$ such that there is a $2^k\times k$ rectangle $R$ in $M$ such that $R$ contains all Boolean strings of length $k$ as rows.
\end{definition}

The VC-dimension in turn characterises the number of examples needed to PAC-learn the concept class given by the rows of the communication matrix of $f$, under any distribution on the columns. Usually in learning theory a concept class is a set of Boolean functions ('concepts'), and here we view rows of the communication matrix of $f$ as functions $f_x(y)=f(x,y)$. The task of PAC learning is for the learner to be able to compute $f_x(y)$ for most $y$ under a distribution $\mu$, after having seen labelled examples from the same distribution. It is well known, that $O(VC(f)\cdot1/\epsilon\cdot \log (1/\epsilon))$ examples suffice \cite{kearns&vazirani:learn}.

 Kremer et al.~\cite{knr:rand1round} proved the following upper bound on one-way communication complexity:   $R^{A\to B,I=0}_\epsilon(f)\leq O(VC(f)\cdot1/\epsilon\cdot\log(1/\epsilon))$. The idea is that Alice and Bob can choose examples $y'$ from the public coin, which Alice can label by sending $f(x,y')$. Bob simulates the PAC learning algorithm  for the rows of the communication matrix, and hence he can successfully predict $f(x,y)$ for most $y$, including (likely) his own input. Note that there is also a lower bound of $Q^{A\to B,I=0}(f)\geq (1-H(\epsilon))VC(f)$ (which is even true in the entanglement assisted case with an additional factor of 1/2)\cite{knr:rand1round,ambainis:racj,klauck:qpcom}.

While it is known, that the number of examples needed to PAC-learn is at least $\Omega(VC(f)/\epsilon)$ \cite{kearns&vazirani:learn}, we get an exponentially better dependence on the error here for the one-way communication model under product distributions.

Our result has an appealing interpretation. Both the one-way model under product distributions and the PAC model can be viewed as learning models (for this it is crucial that the distributional one-way model is considered under product distributions). In the PAC model Alice (or nature) labels random examples drawn from a distribution, and Bob has to end up being able to label new examples mostly correct (under the same, unknown distribution). In the one-way model, there is a known distribution on examples (columns), and a known distribution on concepts (rows). The one-way model under product distributions can clearly simulate any PAC algorithm.
But Alice can send any information she deems useful, not just label examples. Nevertheless, in both models the complexity is determined by the VC-dimension. Is a teacher like Alice not more useful than random labelled examples? We show that the one-way model (i.e., a teacher) is better in the sense that making the error small is exponentially cheaper there, compared to the PAC model.

\begin{theorem}
For all total $f$: $R^{A\to B, I=0}_\epsilon(f)\leq O(Q^{A\to B, I=0}_{1/3}(f)\cdot\log(1/\epsilon))$

\end{theorem}

\begin{proof}
First, $Q^{A\to B, I=0}_{1/3}(f)=\Theta(VC(f)).$ Hence we need to show only that $R^{A\to B,I=0}_\epsilon(f)\leq O(VC(f)\cdot\log(1/\epsilon))$.

For a given distribution $\mu$ on the columns, an $\epsilon$-net among the rows of the communication matrix is a subset $N$ of the set of rows, such that for every row
$x$ there is a row $x'\in N$ which coincides with $x$ with probability $1-\epsilon$ under $\mu$.
We have the following simple observation, due to the fact that Alice can simply send the name of the closest $x'\in N$ to  Bob.

\begin{lemma}
$R^{\mu_A\times \mu_B}_\epsilon(f)$ is upper bounded by the logarithm of the size of the smallest $\epsilon$-net for $f$ and $\mu_B$.
\end{lemma}

Hence instead of the simulation Alice and Bob can agree on an $\epsilon$-net beforehand, and the size of the $\epsilon$-net determines the
complexity of the protocol. Note that PAC-learners also try to find an $\epsilon$-net, but they are restricted to finding one from random examples.
The size of the constructed $\epsilon$-net is much smaller than the number of examples (this is not surprising, since otherwise the concept is not learned yet). Indeed, Sauer's lemma tells us enough about the size of the $\epsilon$-net, when the specified number of examples have been chosen.

\begin{fact}[Sauer]
Let $M$ be a Boolean matrix with $r$ rows and $c$ columns and VC-dimension $d$. Then $r\leq\Phi(c,d)$, where $\Phi(c,d)=\sum_{i=0,\ldots,d}{c\choose i}\leq d\cdot {c\choose d}$.
\end{fact}

We now state the fundamental result from PAC learning (see Theorem 3.3 in \cite{kearns&vazirani:learn}).
\begin{fact}\label{fac:learn}
Consider any function $f:\{0,1\}\times \{0,1\}^n\to\{0,1\}$.
Assume there is a distribution $\mu$ on $y$'s that does not depend on $x$, and $x$ is also fixed but unknown.
The PAC learner is given $c=O(VC(f)\cdot1/\epsilon\cdot\log(1/\epsilon))$ random examples $y_1,\ldots, y_c$ from the distribution together with labels $\ell_1=f(x,y_1),\ldots, \ell_c=f(x,y_c)$.
If the learner chooses any $x'$ that is consistent with these values (i.e., $f(x',y_i)=\ell_i$ for all $i=1,\ldots,c$), then the probability that $f(x',y)\neq f(x,y)$ is at most
$\epsilon$ under $\mu$. Hence, if we choose a string $x'$ consistent with any vector $\ell_1,\ldots, \ell_c$, then we get an $\epsilon$-net for $f,\mu$.
\end{fact}

The size of this $\epsilon$-net is clearly at most $2^c$.
Sauer's lemma can now be used to show that the constructed $\epsilon$-net can be made much smaller.
The size of the $\epsilon$-net constructed in Fact \ref{fac:learn}, without repetitions, is at most the size of the set of {\em distinct} rows
in the matrix for $f$, when we restrict the matrix to the $c$ chosen columns (we may choose one $x'$ for every distinct value of the $c$ labels appearing and add it into the $\epsilon$-net).

The size of the number of distinct rows is bounded now by Sauer's lemma as follows: $VC(f)\cdot {c\choose VC(f)}= VC(f)\cdot {const\cdot VC(f)\cdot1/\epsilon\cdot\log(1/\epsilon)\choose VC(f)}\leq
(1/\epsilon)^{O(VC(f))}$.
Hence the communication is at most the logarithm of this size, which yields the theorem.
\end{proof}

\section{Open Problems}

\begin{itemize}
\item Can the error dependence of a tight upper bound on $Q^{I=0}_\epsilon(DISJ)$ be improved to $\log(1/\epsilon)$?
\item Can the error dependence of $R^{I=0}_\epsilon(f)$ be improved to $\log(1/\epsilon)$ for {\em every} total function $f$?
\item What is the trade-off between the number of rounds and the randomised complexity of DISJ under product distributions?
\item What is the quantum communication complexity of DISJ where the inputs are sets of size $\sqrt n$ from a size $n$ universe? The best known lower bound is $\Omega(n^{1/4})$, the best known upper bound is $O(n^{1/4}\log n)$.
\item What is the largest gap between $Q^{I=0}(f)$ and $R^{I=0}(f)$? In the one-way model there is at most a constant gap for any total function. We have shown a quadratic gap for DISJ.
\end{itemize}

\bibliography{qc}

\section*{Appendix}

\section{Randomised Protocol for DISJ under Product Distributions}

\begin{proof}[Proof of Theorem \ref{thm:prot1}]
Fix any product distribution $\mu$ on $\{0,1\}^n\times\{0,1\}^n$. The main idea is (just like in \cite{bfs:classes}) to have a first phase in which large sets are reduced in size until both sets have size $O(\sqrt n)$. In phase 2 we employ the randomised protocol for DISJ on small sets given by Hastad and Wigderson \cite{hastad:disj} (instead of communicating the sets). To simplify our presentation we describe a randomised protocol.

Set $S=\sqrt n$.
In phase 1 Alice and Bob try to shrink the universe $U$ (without removing positions in $x\cap y$) until the size of $U$ is at most $S$. At that point also $|x\cap U|$ and $|y\cap U|$ are at most of size $S$ and the players move to phase 2. The protocol starts with the universe $U_0=\{1,\ldots,n\}$. The players maintain a current universe $U_i$ until $U_i$ is small at some point.

The protocol proceeds in rounds during phase 1 (we later explain how to get rid of all but two rounds). In each round Alice and Bob exchange a bit each, indicating whether $|x|,|y|\geq S$ or not. If both are smaller, they move to phase 2. The players also maintain a current rectangle of inputs $R_i=A_i\times B_i$ (this would be immediate in a deterministic protocol, but needs to be maintained in the randomised case).

After this exchange, Alice and Bob each compute $Prob(x,y$ are disjoint$)$ on the current distribution restricted to $R_i$ and their row/column. If this probability is less than $\epsilon$ for someone, they reject and quit the protocol.
Otherwise, one player who has a large set still, say Alice, uses the public coin to generate samples $y'\in B_i$. These are disjoint from $x$ with probability at least $\epsilon$. Hence, Alice can name a disjoint $y'$ with expected communication $O(\log (1/\epsilon))$. Since $x\cap y$ is disjoint with $y'$ they set $U_{i+1}=U_i-y'$. The size of the universe decreases by at least $\sqrt n$ in each round in phase 1, the communication is expected $O(\log (1/\epsilon))$ per round, and there are at most $\sqrt n$ rounds.

Phase 2, as mentioned, is the protocol from \cite{hastad:disj}, which solves DISJ with communication $O(\sqrt n\log(1/\epsilon))$ and worst case error $\epsilon$ on sets of size at most $\sqrt n$.

Hence the total expected communication is at most $O(\sqrt n\log(1/\epsilon))$. We need a protocol with a {\em worst case} communication bound, though, but note that during each round in phase 1, using the public coin to pick a new $y'$ corresponds to a Bernoulli trial with success probability at least $\epsilon$.
The communication cost is the logarithm of the number of the first successful trial. The probability that this is larger than $t\log(1/\epsilon)$ is at most
$e^{-1/\epsilon^{t-1}}$. Assume there are $T$ rounds in phase 1. The probability that the message length in any round is more than $(T+1)\log(1/\epsilon)$ is at most $T\cdot e^{-1/\epsilon^T}\leq\epsilon$. Hence we can assume that the message length is at most $(T+1)\log(1/\epsilon)$ in all rounds (the probability that this is not the case is bounded by $\epsilon$).

We now bound the probability that the total message length is more than $10T\log(1/\epsilon)$, by appealing to the Hoeffding bound. Note that the message lengths of all rounds are (still) independent, and that we just established an upper bound on the message length. The Hoeffding bound now implies that the probability of the total message length being larger than the stated bound is at most $\epsilon$. Furthermore, we have that $T\leq \sqrt n$ with certainty. This shows that the communication of phase 1 is at most $O(\sqrt n\log(1/\epsilon)).$ Note that the protocol needs to be modified such that it aborts if the communication in phase 1 exceeds this bound. This introduces error at most $\epsilon$.
\end{proof}

\section{Randomised Protocol and Distributions with Bounded Mutual Information}

\begin{proof}[Proof of Theorem \ref{thm:prot2}]
Fix any distribution $\mu'$ that has information at most $k$.
The protocol we describe again has 2 phases. Informally, the first phase shrinks the sets of Alice and Bob (which could be arbitrarily large) until their sizes are both small enough. The second phase is small set disjointness, as considered before by Hastad and Wigderson \cite{hastad:disj}, and more recently by Saglam and Tardos \cite{tardos:disj}.
We will establish an upper bound of $O(\sqrt{n(k+1)}/\epsilon)$ on the {\em expected} communication complexity with error $\epsilon$. Then the theorem (which claims a worst-case bound) follows via the Markov inequality: if the stated communication bound is violated, stop the protocol and output a random bit.

Set $S=\sqrt{(k+1)n}$. The goal of the first phase is to make both sets smaller than $S$.
Suppose Alice holds $x$ and Bob $y$. They communicate to determine one of them has a set larger than $S$. This needs communication $O(1)$.
If both sets are small we move to phase 2 described below.

In phase 1 Alice and Bob try to shrink the universe $U$ until the size of $U$ is at most $S$. At that point also $|x\cap U|$ and $|y\cap U|$ are at most $S$ and the players move to phase 2. The protocol starts with the universe $U_0=\{1,\ldots,n\}$. The players maintain a current universe $U_i$ until $U_i$ is small at some point.

Note that while the information under $\mu_0=\mu$ is at most $k$, in some branches of the protocol the information on the current sub-rectangle can grow, and we need
that on average it is bounded by $k$. We keep a transcript $T_i=A_i,B_i,C_i$, where $C_i$ is the public coin randomness, and $R_i=A_i\times B _i$ is the rectangle in the communication
matrix induced by the messages up to round $i$, for public coin value $C_i$. Note that the rectangles $R_i$ form a partition of the communication matrix if we fix $C_i$ (since they result from a deterministic protocol then).
By Lemma \ref{lem:decrease} we have that $I(X:Y|T_i)\leq I(X:Y)$.

Denote by $\mu_{t_i}$ the distribution on inputs conditional on the transcript being $T_i=t_i$.
$\mu_{t_i}^x$ is $\mu_{t_i}$ restricted to the row $X=x$. $\tilde{\mu},\tilde{\mu}_{t_i},\tilde{\mu}_{t_i}^x$ denote the distributions restricted to 1-inputs of DISJ.
$\mu_{t_i,Y}$ is the marginal of $\mu_{t_i}$ on Bob's inputs.
$\tilde{\mu}_{t_i,Y}^x$ is the distribution on $y$'s under $T_i=t_i$, for fixed $x$ and conditioned on $x\cap y=\emptyset$.
$\mu_{t_i,Y}^x$ is the distribution on $y$'s under $T_i=t_i$, for fixed $x$.

Here is the protocol for phase 1. Explanations follow.

 \begin{enumerate}
 \item Alice and Bob check whether $|x|\leq S$ and $|y|\leq S$ on $U_i$. If both are, they move to phase 2.
W.l.o.g.~assume that $|y|\geq  S$, otherwise the following steps are done by Bob in an analogous fashion.
\item Alice computes the probability that $DISJ(x,y')=1$ if $y'$ is chosen from $\mu_{t_i}^x$. If this probability is less than $\epsilon/2$, she ends the protocol with output 0.
\item Alice computes $\tilde{\mu}_{t_i,Y}^x$. Another distribution, this one known to both players, is $\mu_{t_i,Y}$.
\item Alice and Bob use rejection sampling as in Fact \ref{fac:sample} (using the distributions $\tilde{\mu}_{t_i,Y}^x$ and $\mu_{t_i,Y}$) to discover a $y'_i$ distributed according to $\tilde{\mu}_{t_i,Y}^x$.
\item Alice and Bob set $U_{i+1}=U_i-y'_i$.
\item $t_{i+1}$ is $t_i$ together with the message and randomness from 1. $\mu_{t_{i+1}}$ is $\mu$ conditioned on $T_{i+1}=t_{i+1}$.
\item Move to step 1.
 \end{enumerate}

 We note the following on the different steps.
 \begin{enumerate}
 \item Communication is $O(1)$.
 \item Clearly the total error introduced by these steps under $\mu$ can never be more than $\epsilon/2$. If the protocol moves ahead the probability of $DISJ(x,y)=1$ is at least $\epsilon/2$ under $\mu_{t_i}^x$.
 \item Since $I(X:Y|T_i)\leq k$ we have that $E_{t_i,x} D(\mu_{t_i,Y}^x||\mu_{t_i,Y})\leq k$.
 \item $D(\tilde{\mu}_{t_i,Y}^x||\mu_{t_i,Y})\leq 2(D(\mu_{t_i,Y}^x||\mu_{t_i,Y})+1)/\epsilon-\log(\epsilon/2)$  due to Lemma \ref{lem:rest} and hence
  the rejection sampling protocol from Fact \ref{fac:sample} uses expected communication $O((k+1)/\epsilon)$. Drawn $y'_i$ are always disjoint from $x$.
\item $|y'_i|\geq S$. Hence $|U_i-U_{i+1}|\geq S$. This step can be performed at most $n/\sqrt{n/(k+1)}$ times.
 \end{enumerate}

 The protocol ends phase 1 with sets $x\cap U_j$ held by Alice and $y\cap U_j$ held by Bob, and $|x\cap U_j|,|y\cap U_j|\leq S$, and $DISJ(x,y)=1\Leftrightarrow DISJ(x \cap U_j,y\cap U_j)=1$. The probability that the protocol ends during phase 1 and makes an error is at most $\epsilon/2$. The expected communication is at most $O(\sqrt{n(k+1)}/\epsilon$.

  Phase 2 is simply the Hastad Wigderson protocol for small set disjointness \cite{hastad:disj}, that finishes the protocol in communication $O(\sqrt{n(k+1)}\log(1/\epsilon))$ and with worst case error $\epsilon/2$. Hence we get a protocol with error $\epsilon$, and expected  communication $O(\sqrt{n(k+1)}/\epsilon)$.

\end{proof}

\section{Randomised Lower Bound for DISJ}

We first bound the information. Letting $X$ and $Y$ follow the marginal distributions of $\mu_{n,k}$, respectively, we have:
\begin{align*}
I(X:Y)&=H(X)-H(X|Y)=\log{n\choose m}-{\bf E}_{y\in Y}(Pr(y)H(X|Y=y))\\
&=\log{n\choose m}-H(X|Y=y_{0})\textrm{ (where $y_{0}$ is any set with $P(y_{0})>0$})\\
&=\log{n\choose m}-\left(2\log{n-m\choose m}+2\log{n-m\choose m-1}+2\right)\\
&\leq\log{n\choose m}-\log{n-m\choose m}=\log\frac{n(n-1)\cdot\ldots\cdot(n-m+1)}{(n-m)\cdot\ldots\cdot(n-2m+1)}\\
&\leq\log\left(1+\frac{m}{n-2m+1}\right)^{m}\leq (\log e)\frac{m^{2}}{n-2m+1}\\
&=c^{2}(\log e)(1+o(1))(k+1)\leq k
\end{align*}
for any sufficiently large $n$.

\begin{proof}[Proof of Theorem \ref{thm:low}]
We may assume that $k=o(n)$, since otherwise (if $k=\Omega(n)$), the original proof by Razborov \cite{razborov:disj} applies directly. Let $l\in\mathbb{N}$ be given and assume that $n=4l-1$. Let $\gamma=\log_{l}(c\sqrt{n(k+1)})$, where $c=(\log e)^{-1}$. Thus $\gamma\in \left(\frac{1}{2},1\right)$ (for $n$ sufficiently large) and our distribution will pick sets of size $l^{\gamma}=c\sqrt{n(k+1)}$. Throughout the proof we will treat numbers like $l^{\gamma}$ as natural numbers, and avoid using the floor function for the sake of readability. We will also identify $\mathcal{P}(\{1,\ldots,n\})$ with $\{0,1\}^{n}$.

We now give an alternative definition for the distribution $\mu=\mu_{n,k}$, as the distribution induced by the following process: First, a triple $T=(T_{1},T_{2},i)$ is chosen uniformly among all such triples, where $\vert T_{1}\vert=\vert T_{2}\vert=2l-1$ and $\{T_{1},T_{2},\{i\}\}$ form a partition of the set $\{1,\ldots,n\}$. Then, with probability $\frac{1}{2}$ the set $x$ is chosen uniformly among all subsets of $T_{1}\cup\{i\}$ with $l^{\gamma}$ elements and such that they contain $i$, and with probability $\frac{1}{2}$ the set $x$ is chosen as a subset of $T_{1}$ with $l^{\gamma}$ elements, again uniformly among all such subsets of $T_{1}$. Similarly, and independently of the choice of $x$, $y$ is chosen with probability $\frac{1}{2}$ uniformly as a subset of $T_{2}\cup\{i\}$ with $l^{\gamma}$ elements and such that it contains $i$, and with probability $\frac{1}{2}$ uniformly among the subsets of $T_{2}$ with $l^{\gamma}$ elements (not containing $i$). Thus non-zero probabilities are assigned only on the set $\{(x,y)\mid x,y\subseteq\{1,\ldots,n\},\vert x\vert=\vert y\vert=l^{\gamma},\vert x\cap y\vert\in\{0,1\}\}$.

Now the statement that $D^{\mu_{n,k}}_\epsilon(DISJ)=\Omega(\sqrt{n(k+1)})$ for any sufficiently small constant $\epsilon>0$, follows directly from Lemma \ref{lem:1} below.
\end{proof}

\begin{lemma}\label{lem:1}
Let $\gamma$ and $\mu$ be defined as in the proof of Theorem \ref{thm:low}. Let $A=\{(x,y)\mid \mu(x,y)>0\textrm{\ and\ }x\cap y=\emptyset\}$ and $B=\{(x,y)\mid \mu(x,y)>0\textrm{\ and\ }x\cap y\neq\emptyset\}$. For any sufficiently small $\epsilon>0$ we have for any rectangle $R=C\times D\subseteq\mathcal{P}(\{1,\ldots,n\})^{2}$ that
\begin{displaymath}
\mu(B\cap R)\geq\Omega(\mu(A\cap R))-2^{-\Omega(n^{\gamma})}.
\end{displaymath}\end{lemma}
\begin{proof}
We consider $\epsilon>0$ to be fixed (but will specify its value later). We begin by defining for any triple $T=(T_{1},T_{2},\{i\})$ as above, the numbers $Row(T)=Pr[x\in C\mid x\subseteq T_{1}\cup\{i\}]$, $Row_{0}(T)=Pr[x\in C\mid x\subseteq T_{1}\cup\{i\},i\notin x]$ and $Row_{1}(T)=Pr[x\in C\mid x\subseteq T_{1}\cup\{i\},i\notin x]$, and similarly $Col(T)=Pr[y\in D\mid y\subseteq T_{2}\cup\{i\}]$, $Col_{0}(T)=Pr[y\in D\mid y\subseteq T_{2}\cup\{i\},i\notin y]$ and $Col_{1}(T)=Pr[y\in D\mid y\subseteq T_{2}\cup\{i\},i\notin y]$. It is important to note that $Row(T)=\frac{1}{2}(Row_{0}(T)+Row_{1}(T))$ and $Col(T)=\frac{1}{2}(Col_{0}(T)+Col_{1}(T))$, just as in the case of Razborov's original distribution, and for the same reasons.

Next, for a triple $T=(T_{1},T_{2},\{i\})$ (and under the above distribution $\mu$) we say that $T$ is \emph{$x$-bad} if $Row_{1}(T)<\frac{1}{6}Row_{0}(T)-2^{-\epsilon n^{\gamma}}$, and that $T$ is \emph{$y$-bad} if $Col_{0}(T)<\frac{1}{6}Col_{0}(T)-2^{-\epsilon n^{\gamma}}$. If $T$ is $x$-bad or $y$-bad, we say that $T$ is \emph{bad}. Let $Bad_{x}(T)$, $Bad_{y}(T)$ and $Bad(T)$ be the respective event indicators.

\begin{claim} For all $t_{2}\subseteq\{1,\ldots,n\}$, with $\vert t_{2}\vert=2l-1$, we have that $Pr[Bad_{x}(T)=1\mid T_{2}=t_{2}]\leq\frac{1}{5}$ and $Pr[Bad_{y}(T)=1\mid T_{2}=t_{2}]\leq\frac{1}{5}$.\end{claim}

\emph{Proof of the Claim.} We prove the first statement, the second one having an almost identical proof.

Let $t_{2}\subseteq\{1,\ldots,n\}$, with $\vert t_{2}\vert=2l-1$, be fixed. Under our distribution, $Row(T)$ can take different values even when $T$ is restricted to partitions for which $T_{2}=t_{2}$. Thus we first treat the case when $\max\{Row(T)\mid T_{2}=t_{2}\}\leq 2^{-\epsilon n^{\gamma}}$. If this inequality holds, then for all $T$ with $T_{2}=t_{2}$ we have: $Row(T)\leq 2^{-\epsilon n^{\gamma}}$, and hence $Row_{0}(T)\leq 2 Row(T)\leq 2\cdot 2^{-\epsilon n^{\gamma}}$ so that $\frac{Row_{0}(T)}{6}-2^{-\epsilon n^{\gamma}}<0\leq Row_{1}(T)$ holds trivially (and hence $Pr[Bad_{x}(T)=1\mid T_{2}]=0$).

Next we treat the case where $\max\{Row(T)\mid T_{2}=t_{2}\}> 2^{-\epsilon n^{\gamma}}$. Define $S=\{x\in C\mid \vert x\vert=l^{\gamma}, x\subset\{1,\ldots,n\}\setminus t_{2}\}$. Note that for any $T$ with $T_{2}=t_{2}$, $Row(T)$ measures the conditional probability (conditioned on $T$) of the same set $S$, with each $x\in S$ having a different (conditional) probability depending on whether $i\in x$. Specifically, if $i\in x$ then the probability of $x$ being chosen, conditioned on $T$, is $\frac{1}{2}{2l-1\choose l^{\gamma}-1}^{-1}$, otherwise the probability is $\frac{1}{2}{2l-1\choose l^{\gamma}}^{-1}=\frac{1}{2}{2l-1\choose l^{\gamma}-1}^{-1}\frac{l^{\gamma}}{2l-l^{\gamma}}=\frac{1}{2}{2l-1\choose l^{\gamma}-1}^{-1}\frac{1}{2l^{1-\gamma}-1}$. Thus, when $T$ is fixed, the probability of each set $x$ containing $i$ is $2l^{1-\gamma}-1$ times that of a set which does not contain $i$.

The proof of this case will proceed as follows: First, we show that under the assumption that a sufficiently large part of the partitions $T$ with $T_{2}=t_{2}$ are $x$-bad, three quarters of the elements of $S$ (which are subsets of $\{1,\ldots,n\}\setminus T_{2}$) must have at least $\frac{21}{25}$ of their elements in a subset of $\{1,\ldots,n\}\setminus T_{2}$ of size $\frac{8l}{5}$. We will then upper-bound the number of subsets of $\{1,\ldots,n\}\setminus T_{2}$ of size $l^{\gamma}$ that have this property (regardless of whether they are in $C$ or not). Next, we will lower-bound $\frac{3}{4}\vert S\vert$ in terms of $\epsilon$, and show that for a suitable choice of $\epsilon$, the lower bound for $\frac{3}{4}\vert S\vert$ is in fact larger than the upper bound we computed before, which is a contradiction showing that it is not possible for that $T$ with $T_{2}=t_{2}$ to be $x$-bad for that many choices of $i$.

Note first that whenever $T_{2}$ is fixed (in our case to $t_{2}$), the choice of $i\in\{1,\ldots,n\}\setminus T_{2}$ also fixes $T_{1}$ and hence all of $T$, and that the choice of $i$ determines the proportion of $x\in S$ whose weights are counted in $Row_{1}(T)$. If for a particular choice of $i$ the resulting $T$ is $x$-bad, then by definition we have that $Row_{1}(T)<\frac{1}{6}Row_{0}(T)-2^{-\epsilon n^{\gamma}}$, and in particular that $Row_{1}(T)<\frac{1}{6}Row_{0}(T)$. If we let $S'$ be the set of $x\in S$ with $i\in x$, then we may rewrite this inequality as:
\begin{align*}
\frac{\vert S'\vert}{{2l-1\choose l^{\gamma}-1}}<\frac{\vert S\vert-\vert S'\vert}{6{2l-1\choose l^{\gamma}}}&\Longleftrightarrow\vert S'\vert<\frac{\vert S\vert-\vert S'\vert}{6(2l^{1-\gamma}-1)}\\
&\Longleftrightarrow\vert S'\vert\left(1+\frac{1}{6(2l^{1-\gamma}-1)}\right)<\frac{\vert S\vert}{6(2l^{1-\gamma}-1)},
\end{align*}
and we may conclude that for $l$ sufficiently large, $\vert S'\vert<\frac{\vert S\vert}{10l^{1-\gamma}}$ (under the assumption that $T$ is $x$-bad). For the last inequality we have used the fact that $\lim_{n\to\infty} l^{1-\gamma}=\infty$, which holds because: $\lim_{n\to\infty}\log l^{1-\gamma}=\lim_{n\to\infty}(1-\gamma)\log l=\lim_{n\to\infty}(1-\log_{l}(c\sqrt{n(k+1)}))\log l\geq\lim_{n\to\infty}(\log l-\log\sqrt{n(k+1)})\geq\lim_{n\to\infty}\log\sqrt\frac{n+1}{16(k+1)}=\infty$ (since $k=o(n)$).

Let $B=\left\{i\in\{1,\ldots,n\}\mid\textrm{the\ partition\ }\left(\{1,\ldots,n\}\setminus (t_{2}\cup\{i\}),t_{2},\{i\}\right)\textrm{ is $x$-bad}\right\}$, and assume that $\vert B\vert\geq\frac{2l}{5}$, that is, assume that for at least one fifth of the possible choices for $i$ the corresponding partition is $x$-bad. By excluding some elements of $B$, we may assume that $\vert B\vert=\frac{2l}{5}$. Now, if we consider the number of pairs $(x, i)$ with $x\in S$ and $i\in x$, we have by the inequality in the last paragraph that each of the $i\in B$ can be the second element of at most $\frac{\vert S\vert}{10l^{1-\gamma}}$ such pairs, and hence $B$ can contribute the second element of at most $\frac{2l}{5}\frac{\vert S\vert}{10l^{1-\gamma}}=\frac{1}{25}l^{\gamma}\vert S\vert$ of the total of $l^{\gamma}\vert S\vert$ pairs. Applying the Colouring Lemma below with $X=S$, $Y=\{1,\ldots,l^{\gamma}\}$, $c(x,i)=0$ if and only if the $i$-th smallest element of $x$ is in $B$ (so that $p\geq\frac{24}{25}$) and $r=\frac{21}{25}$, we have that at least three quarters of all $x\in S$ have the property that more than $\frac{21}{25}$ of their elements lie in $G=\{1,\ldots,n\}\setminus (t_{2}\cup B)$. Let $Q$ be the set of subsets $x\subseteq B\cup G=\{1,\ldots,n\}\setminus t_{2}$, with $\vert x\vert=l^{\gamma}$ and the property that $\vert x\cap G\vert\geq\frac{21}{25}l^{\gamma}$. Then we must have that $\vert Q\vert\geq\frac{3}{4}\vert S\vert$. We will now upper-bound the size of the set $Q$.

Since every $x\in Q$ can have a proportion of at most $4/25$ of its elements in $B$, we have that
\begin{align*}
\log\vert Q\vert&\leq\log\left[\sum^{\frac{4}{25}l^{\gamma}}_{i=0}{\frac{2l}{5}\choose i}{\frac{8l}{5}\choose l^{\gamma}-i}\right]\leq\log\left[\sum^{\frac{4}{25}l^{\gamma}}_{i=0}\left(\frac{2le}{5i}\right)^{i}\left(\frac{8le}{5(l^{\gamma}-i)}\right)^{l^{\gamma}-i}\right]\\
&\leq\log\left[\frac{4}{25}l^{\gamma}\left(\frac{2le}{5}\frac{25}{4l^{\gamma}}\right)^{\frac{4}{25}l^{\gamma}}\left(\frac{8le}{5}\frac{25}{21l^{\gamma}}\right)^{\frac{21}{25}l^{\gamma}}\right]\\
&=\log\left[\frac{4}{25}l^{\gamma}\left(\frac{5e}{2}l^{1-\gamma}\right)^{\frac{4}{25}l^{\gamma}}\left(\frac{40e}{21}l^{1-\gamma}\right)^{\frac{21}{25}l^{\gamma}}\right]\\
&\leq\gamma\log l+\frac{4}{25}l^{\gamma}\log\left(\frac{5e}{2}l^{1-\gamma}\right)+\frac{21}{25}l^{\gamma}\log\left(\frac{40e}{21}l^{1-\gamma}\right)+O(1)\\
&=(1-\gamma)l^{\gamma}\log l+\left(\frac{4}{25}\log\frac{5e}{2}+\frac{21}{25}\log\frac{40e}{21}\right)l^{\gamma}+O(\log l)\\
&\leq(1-\gamma)l^{\gamma}\log l+2.43508\cdot l^{\gamma}+O(\log l),
\end{align*}
where in the first line we used the inequality ${m\choose k}\leq\left(\frac{em}{k}\right)^{k}$ for each term of the sum. The inequality sign between the first and second line can be justified as follows: For $x\in(0,\frac{1}{2})$, consider the expression
\begin{displaymath}
\log\left[\left(\frac{\frac{2}{5}el}{x\cdot l^{\gamma}}\right)^{x\cdot l^{\gamma}}\left(\frac{\frac{8}{5}el}{(1-x)l^{\gamma}}\right)^{(1-x)l^{\gamma}}\right]=(1-\gamma)l^{\gamma}\log l+l^{\gamma}\left(x\log\frac{2e}{5x}+(1-x)\log\frac{8e}{5(1-x)}\right)
\end{displaymath}
and set $f(x)=x\log\frac{2e}{5x}+(1-x)\log\frac{8e}{5(1-x)}=x\log\frac{1}{x}+(1-x)\log\frac{1}{1-x}+(1+\log e)x+(3+\log e)(1-x)-\log 5=3+\log e-\log 5+H(x)-2x$. Then $f'(x)=H'(x)-2=\log\frac{1-x}{x}-2$. Note that the function $\frac{1-x}{x}$ is decreasing but positive on $(0,1)$, and we have that the smallest value $x_{0}\in (0,\frac{1}{2})$ for which we can have $f'(x_{0})=0$ is $x_{0}=\frac{1}{5}$, which implies that $f(x)$, and hence also the argument of the logarithm in the expression above, is strictly increasing on $(0,\frac{1}{5})$. Thus the terms of the sum
\begin{displaymath}
\sum^{\frac{4}{25}l^{\gamma}}_{i=0}\left(\frac{2le}{5i}\right)^{i}\left(\frac{8le}{5(l^{\gamma}-i)}\right)^{l^{\gamma}-i}
\end{displaymath}
are increasing, so that each term is upper-bounded by the final term, which justifies the inequality between the the first and second line above.

Next we compute a lower bound for $\frac{3}{4}\vert S\vert$. Let $T^{*}$ be a partition with $T^{*}_{2}=t_{2}$ and $Row(T^{*})=\max\{Row(T)\mid T_{2}=t_{2}\}$. Then we have that $\frac{3}{4}\vert S\vert=\frac{3}{4}[Row_{0}(T^{*}){2l-1\choose l^{\gamma}}+Row_{1}(T^{*}){2l-1\choose l^{\gamma}-1}]\geq\frac{3}{4}(Row_{0}(T^{*})+Row_{1}(T^{*})){2l-1\choose l^{\gamma}-1}=\frac{6}{4}Row(T^{*}){2l-1\choose l^{\gamma}-1}>2^{-\epsilon n^{\gamma}}{2l-1\choose l^{\gamma}-1}$. Finally we have:
\begin{align*}
\log\frac{3}{4}\vert S\vert&>\log\left[2^{-\epsilon n^{\gamma}}{2l-1\choose l^{\gamma}-1}\right]\geq l^{\gamma}\log\left((e-o(1))\frac{2l-1}{l^{\gamma}-1}\right)-\epsilon n^{\gamma}-\Theta(\log l)\\
&\geq (1-\gamma)l^{\gamma}\log l+l^{\gamma}\log(2(e-o(1)))-\epsilon\cdot(4l-1)^{\gamma}-\Theta(\log l)\\
\textrm{(for large $l$)}&\geq (1-\gamma)l^{\gamma}\log l+2.4426\cdot l^{\gamma}-\epsilon\cdot(4l)^{\gamma}-\Theta(\log l).
\end{align*}
For $\epsilon\leq\frac{1}{1000\cdot 4^{\gamma}}$ we get the desired contradiction, that $\frac{3}{4}\vert S\vert>\vert Q\vert$.

(The lower-bound for ${2l-1\choose l^{\gamma}-1}$ above can be obtained using the Stirling bounds for the factorial, $\sqrt{2\pi n}\left(\frac{n}{e}\right)^{n}\leq n!\leq e\sqrt{n}\left(\frac{n}{e}\right)^{n}$, as follows:
\begin{align*}
{2l-1\choose l^{\gamma}-1}&\geq\frac{\sqrt{2\pi(2l-1)}\cdot(2l-1)^{2l-1}}{e^{2}\sqrt{(l^{\gamma}-1)(2l-l^{\gamma})}\cdot(l^{\gamma}-1)^{l^{\gamma}-1}\cdot(2l-l^{\gamma})^{2l-l^{\gamma}}}\\
&=\frac{\sqrt{2\pi(2l-1)}}{e^{2}\sqrt{(l^{\gamma}-1)(2l-l^{\gamma})}}\left(\frac{2l-1}{l^{\gamma}-1}\right)^{l^{\gamma}-1}\left(\frac{2l-1}{(2l-1)-(l^{\gamma}-1)}\right)^{2l-l^{\gamma}}\\
&=\frac{\sqrt{2\pi(2l-1)}}{e^{2}\sqrt{(l^{\gamma}-1)(2l-l^{\gamma})}}\left(\frac{2l-1}{l^{\gamma}-1}\right)^{l^{\gamma}-1}\left[\left(1+\frac{l^{\gamma}-1}{2l-l^{\gamma}}\right)^{\frac{2l-l^{\gamma}}{l^{\gamma}-1}}\right]^{l^{\gamma}-1}\\
&\geq\frac{\sqrt{2\pi}}{e^{2}}\frac{1}{\sqrt{(l^{\gamma}-1)}}\left(\frac{2l-1}{l^{\gamma}-1}\right)^{l^{\gamma}-1}(e-o(1))^{l^{\gamma}-1}.)
\end{align*}

\begin{claim} ${\bf E}[Row_{0}(T)Col_{0}(T)(1-Bad(T))]>\frac{1}{5}{\bf E}[Row_{0}(T)Col_{0}(T)]$.\end{claim}

\emph{Proof of the Claim.} Since $Bad(T)\leq Bad_{x}(T)+Bad_{y}(T)$, it is enough to prove that\\ ${\bf E}[Row_{0}(T)Col_{0}(T)Bad_{x}(T)]\leq\frac{2}{5}{\bf E}[Row_{0}(T)Col_{0}(T)]$, with a similar statement for $Bad_{y}(T)$ being proved in the same fashion. For each $t_{2}\subseteq\{1,\ldots,n\}$, with $\vert t_{2}\vert=2l-1$, we will show that the desired inequality holds when conditioned on $T_{2}=t_{2}$, which implies that the unconditioned inequality holds. All triples $T$ with $T_{2}=t_{2}$ have the same value for $Col_{0}(T)$, so let this value be called $c'$. Also let $r={\bf E}[Row(T)\mid T_{2}=t_{2}]$. Now we have:
\begin{align*}
&{\bf E}[Row_{0}(T)Col_{0}(T)Bad_{x}(T)\mid T_{2}=t_{2}]\\
\leq& c'{\bf E}[Row_{0}(T)Bad_{x}(T)\mid T_{2}=t_{2}]\\
\leq&c'{\bf E}\left[2\cdot {\bf E}[Row(T)\mid T_{2}=t_{2}]\cdot Bad_{x}(T)\mid T_{2}=t_{2}\right]\\
\leq& 2 c'r{\bf E}[Bad_{x}(T)\mid T_{2}=t_{2}]\\
\leq&\frac{2}{5}c' r\textrm{\ (by Claim 1)}\\
=&\frac{2}{5}c'{\bf E}[Row(T)\mid T_{2}=t_{2}]\\
=&\frac{2}{5}c'{\bf E}[Row_{0}(T)\mid T_{2}=t_{2}]\\
=&\frac{2}{5}{\bf E}[Row_{0}(T)Col_{0}(T)\mid T_{2}=t_{2}]
\end{align*}
The inequality between the second and the third line can be justified as follows: Recall that, as observed in the proof of Claim 1, even when considering only triples $T_{2}=t_{2}$, the value of $Row(T)$ can differ by a factor of at most $2l^{1-\gamma}-1$. This is due to the fact that $Row(T)$ measures the probability (conditioned on $T$) of the same set $S=\{x\in C\mid \vert x\vert=l^{\gamma},x\subseteq\{1,\ldots,n\}\setminus t_{2}\}$, but depending on whether $i\in x$ (for a particular choice of $i$ and hence of $T$), an $x\in S$ will have (conditional) probability either $\frac{1}{2}{2l-1\choose l^{\gamma}-1}^{-1}$ or $\frac{1}{2}{2l-1\choose l^{\gamma}}^{-1}$. Thus if $T^{*}$ is a triple with $T^{*}_{2}=t_{2}$ for which $Row_{0}(T^{*})=\max\{Row_{0}(T)\mid T_{2}=t_{2}\}$, then $Row(T^{*})$ must be the minimum among all values of $Row(T)$ when $T_{2}=t_{2}$, because when $T=T^{*}$ the largest portion of elements of $S$ have probability $\frac{1}{2}{2l-1\choose l^{\gamma}}^{-1}$ instead of $\frac{1}{2}{2l-1\choose l^{\gamma}-1}^{-1}$. It follows that for all $T$ with $T_{2}=t_{2}$ we have $Row(T)\geq Row(T^{*})\geq\frac{1}{2}Row_{0}(T^{*})$, and hence that ${\bf E}[2Row(T)\mid T_{2}=t_{2}]\geq Row_{0}(T^{*})$. On the other hand we have that for all $T$ with $T_{2}=t_{2}$, $Row_{0}(T)\leq Row_{0}(T^{*})$, so finally we get that $Row_{0}(T)\leq 2{\bf E}[Row(T)\mid T_{2}=t_{2}]$ for all $T$ with $T_{2}=t_{2}$. $\blacksquare$

\begin{claim} For any rectangle $R$: $\mu(B\cap R)=\frac{1}{4}{\bf E}[Row_{1}(T)Col_{1}(T)]$ and $\mu(A\cap R)=\frac{3}{4}{\bf E}[Row_{0}(T)Col_{0}(T)]$ (with the expectation taken over all partitions $T$).\end{claim}

The proof of this claim is identical to the case where $\mu$ is the distribution in Razborov's proof (see \cite{kushilevitz&nisan:cc}), since the relevant observations also apply to our modified distribution: 1. $\mu(B)=\frac{1}{4}$ (and hence $\mu(A)=\frac{3}{4}$), because for every fixed partition $T$, $i\in x$ with probability $\frac{1}{2}$ and $i\in y$ with probability $\frac{1}{2}$, independently. 2. $i\in x$ and $i\in y$ are independent events (for the same reason). 3. For every $(x,y)$ with $x\cap y=\emptyset$ we have that $Pr[(x,y)\mid (i\notin x)\wedge(i\notin y)]=Pr[(x,y)\mid ((i\notin x)\wedge(i\notin y))\vee((i\in x)\wedge(i\notin y))\vee((i\notin x)\wedge(i\in y))]$, because conditioning on either one of the two events induces the uniform distribution on the set $\left\{(x,y)\mid x,y\subset\{1,\ldots,n\},x\cap y=\emptyset,\vert x\vert=\vert y\vert=l^{\gamma}\right\}$.

We now use claims 2 and 3 to prove the statement of the lemma:
\begin{align*}
&\mu(B\cap R)=\frac{1}{4}{\bf E}[Row_{1}(T)Col_{1}(T)]\\
\geq&\frac{1}{4}{\bf E}[Row_{1}(T)Col_{1}(T)(1-Bad(T))]\\
\geq&\frac{1}{4}{\bf E}\left[\left(\frac{Row_{0}(T)}{6}-2^{-\epsilon n^{\gamma}}\right)\left(\frac{Col_{0}(T)}{6}-2^{-\epsilon n^{\gamma}}\right)(1-Bad(T))\right]\textrm{\ (by\ def.\ of\ $Bad$)}\\
=&\frac{1}{4}{\bf E}\left[\left(\frac{Row_{0}(T)Col_{0}(T)}{36}-\frac{2^{-\epsilon n^{\gamma}}}{6}\left(Row_{0}(T)+Col_{0}(T)\right)+2^{-2\epsilon n^{\gamma}}\right)(1-Bad(T))\right]\\
\geq&\Omega\left({\bf E}[Row_{0}(T)Col_{0}(T)(1-Bad(T))]\right)-2^{-\epsilon n^{\gamma}}\textrm{\ (since $Row_{0}(T)+Col_{0}(T)\leq 2$)}\\
\geq&\Omega\left({\bf E}[Row_{0}(T)Col_{0}(T)]\right)-2^{-\epsilon n^{\gamma}}\textrm{\ (by Claim 2)}\\
\geq&\Omega(\mu(A\cap R))-2^{-\epsilon n^{\gamma}}\textrm{\ (by Claim 3)}
\end{align*}
Choosing $\epsilon$ to be smaller than both the constant in front of $\mu(A\cap R)$ and $\frac{1}{1000\cdot 4^{\gamma}}$ completes the proof.
\end{proof}

\begin{lemma}[Colouring Lemma.] Let $X$ and $Y$ be non-empty finite sets, and let $c:X\times Y\mapsto\{0,1\}$ be a colouring of $X\times Y$ such that a proportion $p\in(0,1)$ of the elements of $X\times Y$ are mapped to $1$, that is, such that $\vert c^{-1}(1)\vert /\vert X\times Y\vert=p$. Then for any $r\in(0,p)$ such that $r\vert Y\vert\in\mathbb{N}$, we have that for at least $\frac{p-r}{1-r}\vert X\vert$ elements $x\in X$, $\vert(\{x\}\times Y)\cap c^{-1}(1)\vert>r\vert Y\vert$.
\end{lemma}

\begin{proof}
We call sets of the form $\{x\}\times Y$ \emph{rows}, and let the number $w(x)=\sum_{y\in Y}c(x,y)=\vert(\{x\}\times Y)\cap c^{-1}(1)\vert$ be the \emph{weight} of the row $\{x\}\times Y$, for each $x\in X$. Let $c$ be a colouring of $X\times Y$ as above, but such that the smallest possible proportion of rows have weight $>r\vert Y\vert$, and denote this proportion by $q$. Thus $q$ is such that for any colouring $c'$ satisfying the conditions of the lemma, at least $q\vert X\vert$ elements $x\in X$ satisfy $\vert(\{x\}\times Y)\cap c^{-1}(1)\vert>r\vert Y\vert$.

We may assume that all rows with weight $\leq r\vert Y\vert$ have weight exactly $r\vert Y\vert$: If this is not the case, we may repeatedly perform the operation of changing a $0$ into $1$ on a row with weight $<r\vert Y\vert$, and a $1$ into $0$ on a row with weight $>r\vert Y\vert$, until the above statement is true. (It is easy to see that the colouring $c$ must have rows with weight $>r\vert Y\vert$, since otherwise the overall proportion of elements mapped to $1$ would be $\leq r<p$.) This operation leaves the proportion of elements that are mapped to $1$ unchanged, and the minimality of the chosen colouring $c$ guarantees that the number of rows with weight $>r\vert Y\vert$ does not decrease (and therefore remains unchanged).

Next, we may assume that all but at most one of the rows with weight $>r\vert Y\vert$ have weight exactly $\vert Y\vert$: If this is not the case, we may fix one such row, replace all zeroes with ones on all other rows of weight $>r\vert Y\vert$ (thus making their weight exactly $\vert Y\vert$), and on the fixed row change the same number of ones into zeroes so as to match the changes made on all other rows. Again the overall proportion of elements being mapped to $1$ does not change, and the minimality of the colouring $c$ guarantees that the weight of the fixed row stays $>r\vert Y\vert$.

Based on the above we now have: $p\vert X\vert\vert Y\vert=q\vert X\vert\vert Y\vert-\alpha\vert Y\vert+(1-q)\vert X\vert r\vert Y\vert$, where $\alpha\in [0,1-r)$ is the proportion of zeroes on the one row that has weight $>r\vert Y\vert$ but not necessarily $=\vert Y\vert$. Thus we have:
\begin{align*}
p\leq q+(1-q)r\Longleftrightarrow p\leq (1-r)q+r\Longleftrightarrow\frac{p-r}{1-r}\leq q.
\end{align*}
\end{proof}

\end{document}